\documentclass[11pt]{article}
\usepackage[T1]{fontenc}
\usepackage{fullpage}
\usepackage{graphicx}
\usepackage{amsmath,amssymb,amsthm,amsfonts,dsfont,mathtools}
\usepackage{relsize}
\usepackage[margin=1cm]{caption}
\usepackage{wrapfig}
\usepackage{frame,color}
\usepackage{environ,subfig}
\usepackage{hyperref}
\usepackage{xspace}
\usepackage{framed}
\usepackage{comment}
\usepackage{changepage}

\theoremstyle{plain}
\newtheorem{theorem}{Theorem}[section]
\newtheorem{lemma}[theorem]{Lemma}

\theoremstyle{definition}

\newcommand{\ignore}[1]{}

\newcommand{\Expect}{\operatorname{E}}
\newcommand{\Prob}{\operatorname{Pr}}
\newcommand{\bydef}{\stackrel{\rm def}{=}}

\newcommand{\ceil}[1]{\left\lceil #1 \right\rceil}

\newcommand{\poly}{{\operatorname{poly}}}

\newcommand{\rb}[2]{\raisebox{#1 mm}[0mm][0mm]{#2}}
\newcommand{\istrut}[2][0]{\rule[- #1 mm]{0mm}{#1 mm}\rule{0mm}{#2 mm}}

\newcommand{\listen}{\textsf{listen}}
\newcommand{\send}{\textsf{transmit}}
\newcommand{\idle}{\textsf{idle}}
\newcommand{\noise}{\textsf{noise}}
\newcommand{\silence}{\textsf{silence}}
\newcommand{\leader}{\textsf{leader}}
\newcommand{\nonleader}{\textsf{follower}}

\newcommand{\pp}[2]{p_{{#1}\rightsquigarrow{#2}}}
\newcommand{\Tideal}{\mathcal{T}_{\text{no-comm}}}

\newcommand{\strong}{\textsf{Strong-CD}}
\newcommand{\weak}{\textsf{Sender-CD}}
\newcommand{\recv}{\textsf{Receiver-CD}}
\newcommand{\nocd}{\textsf{No-CD}}

\newcommand{\LeaderElection}{\textsf{Leader Election}}
\newcommand{\ApproximateCounting}{\textsf{Approximate Counting}}
\newcommand{\Census}{\textsf{Census}}
\newcommand{\SuccessfulCommunication}{\textsf{Successful Communication}}
\newcommand{\Szemeredi}{Szemer\'{e}di}
\newcommand{\Erdos}{Erd\H{o}s}
\newcommand{\Renyi}{R\'{e}nyi}
\newcommand{\Sos}{S\'{o}s}

\NewEnviron{fcodebox}[1]%
{\fbox{%
\begin{minipage}{#1\linewidth}
\begin{codebox}
\BODY
\end{codebox}
\end{minipage}
}}

\title{Exponential Separations in the Energy\\ Complexity of Leader Election\thanks{Supported by NSF grants CNS-1318294, CCF-1514383, CCF-1637546, and CCF-1815316.
Research performed while Ruosong Wang and Wei Zhan were visiting University of Michigan.
Ruosong Wang and Wei Zhan are supported in part by the National Basic Research Program of China Grant 2015CB358700, 2011CBA00300, 2011CBA00301, the National Natural Science Foundation of China Grant 61202009, 61033001, 61361136003.}}

\author{Yi-Jun Chang\\ University of Michigan \\ \footnotesize \texttt{cyijun@umich.edu}\and
Tsvi Kopelowitz\\ Bar-Ilan University \\ \footnotesize \texttt{kopelot@gmail.com} \and
Seth Pettie\\ University of Michigan \\ \footnotesize \texttt{pettie@umich.edu} \vspace*{.3cm}
\and
Ruosong Wang\\ Carnegie Mellon University \\ \footnotesize \texttt{ruosongw@andrew.cmu.edu} \and
Wei Zhan\\ Princeton University \\ \footnotesize \texttt{weizhan@cs.princeton.edu}}

\begin{document}
\date{}
\maketitle
\thispagestyle{empty}
\setcounter{page}{0}

\begin{abstract}
{\em Energy} is often the most constrained resource for battery-powered wireless devices,
and most of the energy is often spent on {\em transceiver usage} (i.e., transmitting and receiving packets) rather than computation.
In this paper we study the {\em energy complexity}
of fundamental problems in several models of wireless radio networks.
It turns out that energy complexity is very sensitive to whether the devices can generate random bits
and their ability to {\em detect collisions}.  We consider four collision detection models:
\strong{} (in which transmitters and listeners detect collisions),
\weak{} (in which only transmitters detect collisions) and \recv{} (in which only listeners detect collisions),
and \nocd{} (in which no one detects collisions).

The take-away message of our results is quite surprising. For randomized algorithms,
there is an exponential gap between the energy complexity of \weak{} and \recv:
\[
\mbox{Randomized:} \; \mbox{\nocd} \;=\; \mbox{\weak} \;\gg\; \mbox{\recv} = \mbox{\strong}
\]
and for deterministic algorithms, there is another exponential gap in energy complexity,
\emph{but in the reverse direction}:
\[
\mbox{Deterministic:} \; \mbox{\nocd} \;=\; \mbox{\recv} \;\gg\; \mbox{\weak} = \mbox{\strong}
\]
Precisely, the randomized energy complexity
of \LeaderElection{} is $\Theta(\log^* n)$ in \weak{} but $\Theta(\log(\log^* n))$ in \recv,
where $n$ is the number of devices, which is unknown to the devices at the beginning;
the deterministic complexity of \LeaderElection{} is $\Theta(\log N)$ in \recv{} but $\Theta(\log\log N)$ in \weak,
where $N$ is the size of the ID space.

There is a tradeoff between time and energy.  We provide a new upper bound on the
time-energy tradeoff curve for randomized algorithms.
A critical component of this algorithm is a new deterministic \LeaderElection{} algorithm for {\em dense} instances, when $n=\Theta(N)$, with inverse Ackermann energy complexity.
\end{abstract}
\newpage

\section{Introduction}\label{sect:introduction}

In many networks of wireless devices the scarcest resource is {\em energy},
and the  lion's share of energy is often spent on {\em radio transceiver}
usage~\cite{PolastreSC05,BarnesCMA10,Lee+13,SivalingamSA97}---transmitting and receiving packets---%
not on computation per se.
In this paper we investigate the {\em energy complexity} of fundamental
problems in synchronized single-hop wireless networks: \LeaderElection,
\ApproximateCounting, and taking a \Census.

In all models we consider
{\em time} to be partitioned into discrete slots;
all devices have access to a single shared channel and can choose, in each time slot,
to either \send{} a message $m$ from some space $\mathcal{M}$, \listen{} to the channel, or remain \idle.
Transmitting and listening each cost one unit of energy; we measure the energy usage
of an algorithm on $n$ devices by the worst case energy usage of any device.
For the sake of simplicity we assume computation is free and the message size is unbounded.
If exactly one device transmits, all listeners hear the message $m$, and if zero devices
transmit, all listeners hear a special message $\lambda_S$ indicating \silence.
We consider four collision detection models depending on whether transmitters and listeners can detect collisions.

\begin{description}
\item[\strong.] Each transmitter and listener receives one of three signals:
$\lambda_S$ (\silence, if zero devices transmit), $\lambda_N$ (\noise, if $\ge 2$ devices transmit),
or a message $m \in \mathcal{M}$ (if one device transmits).

\item[\weak.] (Often called ``No-CD''~\cite{JurdzinskiKZ02}) Each transmitter and listener receives one of two signals:
$\lambda_S$ (zero or $\ge 2$ devices transmit), or a message $m \in \mathcal{M}$ (if one device transmits).
Observe that the \weak{} model has no explicit collision detection, but still allows for sneaky collision detection:
if a sender hears $\lambda_S$, it can infer that there was at least one other sender.

\item[\recv.] Transmitters receive no signal.  Each listener receives one of three signals:
$\lambda_S$ (\silence, if zero devices transmit), $\lambda_N$ (\noise, if $\ge 2$ devices transmit),
or a message $m \in \mathcal{M}$ (if one device transmits).

\item[\nocd.] Transmitters receive no signal. Listeners receive one of two signals:
$\lambda_S$ (zero or $\ge 2$ devices transmit) or a message $m \in \mathcal{M}$.
\end{description}



Each of the four models comes in both randomized and deterministic variants.
A key issue is breaking symmetry.  Whereas randomized models easily accomplish
this by having devices flip independent random coins, deterministic models
depend on having pre-assigned unique IDs to break symmetry.

\begin{description}
\item[Randomized Model.] In the randomized
model all $n$ devices begin in exactly the same state,
and can break symmetry by
generating private random bits.
The number $n$ is unknown and unbounded.
The maximum allowed failure probability of a randomized algorithm is at most $1/\poly(n)$.
In a \emph{failed} execution, devices may consume unbounded energy and never halt~\cite{JurdzinskiKZ02b,JurdzinskiKZ02c,Brandes2016}.
\item[Deterministic Model.] All $n$ devices have unique IDs in the range $[N] \bydef \{1,\ldots,N\}$,
where $N$ is common knowledge but $n \leq N$ is unknown.
\end{description}
To avoid impossibilities, in the \nocd{} model it is promised that $n\ge 2$. See Section~\ref{sec:otherlb}
for a discussion of \emph{loneliness detection}.

It could be argued that real world devices rarely endow transmitters with
more collision detection power than receivers, so the \weak{} model does not merit study.
We feel this thinking gets the order backwards.  There is a certain cost for equipping
tiny devices with extra capabilities (e.g., generating random bits or detect collisions)
so how are we to tell whether adding these capabilities is \emph{worth the expense}?
To answer that question we \emph{first} need to
determine the complexity of the problems that will ultimately be solved by the network.
The goal of this work is to understand the power of various abstract models,
not to cleave closely to existing real world technologies, simply because they exist.
In this paper, we consider the following three fundamental distributed problems.

\begin{description}
\item[Leader Election.] Exactly
one device designates itself the \leader{} and all others designate
themselves \nonleader. For technical reasons, we require that the computation ends when the \leader{} sends a message while every \nonleader{} listens to the channel.

\item[Approximate Counting.] At the end of the computation all devices agree on
an estimate $\tilde{n}$ of the network size $n$ such that $\tilde{n}=\Theta(n)$.

\item[Census.]
At the end of the computation some device announces a list
of the IDs of all devices. We only study this problem in the deterministic model.
\end{description}

Notice that any deterministic algorithm that solves \Census{} is also capable of solving \LeaderElection{} and \ApproximateCounting{} with the same runtime and energy cost.

\subsection{New Results}
In the randomized model, we show
that the energy complexity of \LeaderElection{} and \ApproximateCounting{}
are $\Theta(\log^* n)$ in \weak{} and \nocd{} but $\Theta(\log(\log^* n))$ in \strong{} and \recv{}. The lower bounds also apply to the {\em contention resolution} problem, and this establishes that the recent $O(\log(\log^* n))$ contention resolution
protocol of Bender, Kopelowitz, Pettie, and Young~\cite{BenderKPY16} is optimal.
 Our upper bounds  offer a time-energy tradeoff. See Table~\ref{table:time-energy-intro} for the energy cost of our algorithm under different runtime constraints.

\begin{table}[h!]
\centering
{\small
\begin{center}
  \begin{tabular}{| p{4.1cm} | p{4.8cm} | p{5.5cm} | }
    \hline
    	& \multicolumn{2}{c|}{\sc Energy Complexity}\\
\rb{1.7}{\sc Time Complexity} &
\multicolumn{1}{l}{\strong{}  or \recv{}} &
\multicolumn{1}{l|}{\weak{}  or \nocd{}}
    \\ \hline
    $O(n^{o(1)})$
    &$O(\log (\log^\ast n))$
    &$O(\log^\ast n)$\istrut[2]{4}
    \\ \hline
    $O(\log^{2+\epsilon} n)$, $0<\epsilon \leq O(1)$
    &$O(\log ( \epsilon^{-1} \log \log \log n))$
    &$O(\epsilon^{-1} \log \log \log n)$\istrut[2]{4}
    \\ \hline
    $O(\log^{2} n)$
    &$O(\log \log  \log n)$
    &$O(\log \log n)$\istrut[2]{4}\\\hline
  \end{tabular}
\end{center}
}
\caption{\label{table:time-energy-intro} Time-energy tradeoff of randomized \ApproximateCounting{} and \LeaderElection.
Notice that the third line is a special case of the second line when $\epsilon = 1 / \log \log n$.}
\end{table}

For \LeaderElection{} we establish matching bounds in all the deterministic models.
In \strong~and \weak, \LeaderElection\ requires $\Omega(\log\log N)$ energy even when $n=2$,
and \Census\ can be solved with $O(\log \log N)$ energy and $O(N)$ time, for any $n\le N$.
However, in \nocd{} and \recv, the energy complexity of these
problems jumps to $\Theta(\log N)$~\cite{JurdzinskiKZ02c}.

Finally, we prove that when the input is {\em dense} in the ID space,
meaning $n=\Theta(N)$, \Census{} can actually be computed with
only $O(\alpha(N))$ energy and $O(N)$ time, even in \nocd.
To our knowledge, this is the first time inverse-Ackermann-type recursion
has appeared in distributed computing.

\subsection{Prior Work}

Jurdzinski et al.~\cite{JurdzinskiKZ02} studied the deterministic energy complexity of \LeaderElection{}
in the \weak{} model.  They proved that dense instances $n=\Theta(N)$ can be solved
with $O(\log^* N)$ energy, and claimed that the complexity of the sparse instances is between
$\Omega(\log\log N/\log\log\log N)$ and $O(\log^{\epsilon} N)$.
While the lower bound is correct, the algorithm presented in~\cite{JurdzinskiKZ02} is not.\footnote{T. Jurdzinski. (Personal communication, 2016.)}
The most efficient published algorithm uses $O(\sqrt{\log N})$ energy, also due to Jurdzinski et al.~\cite{JurdziskiKZ03}.
The same authors~\cite{JurdzinskiKZ02b} gave a reduction from randomized \weak{} \ApproximateCounting{}
to deterministic \LeaderElection{} over ID space $N=O(\log n)$, which, using~\cite{JurdziskiKZ03},
leads to an $O(\sqrt{\log\log n})$
energy algorithm for \ApproximateCounting.
In~\cite{JurdzinskiKZ02c} the authors gave a method for simulating
\weak{} protocols in the \nocd{} model, and proved that deterministic \nocd{} \LeaderElection{} takes $\Omega(\log N)$ energy.
Nakano and Olariu~\cite{NakanoO00} showed that $n$ devices in the \weak{} model
can assign themselves distinct IDs in $\{1,\ldots,n\}$ with $O(\log\log n)$ energy in expectation.

Recently, Bender et al.~\cite{BenderKPY16} gave a method for
circuit simulation in the \strong{} model,
which led to randomized \ApproximateCounting{} and \LeaderElection{} protocols
using $O(\log(\log^* n))$ energy and $n^{o(1)}$ time.
An earlier algorithm of Kardas et al.~\cite{kardas2013energy} solves \LeaderElection{}
in the \strong{} model in $O(\log^\epsilon n)$ time using $O(\log\log\log n)$ energy,
in expectation but not with high probability.

Most of the previous work in the radio network model has been concerned with {\em time},
not energy.   Willard~\cite{Willard86} proved that $O(\log\log n)$ time is necessary and sufficient
for one device to successfully transmit in the \strong{} model with constant probability;
see~\cite{NakanoO00b} for tradeoffs between time and success probability.
In the \weak{} model this problem requires $\Theta(\log^2 n)$ time to solve, with probability $1-1/\poly(n)$~\cite{Farach-ColtonFM06,JurdzinskiS02,Newport14}.
Greenberg and Winograd~\cite{greenberg1985lower}
proved that if {\em all} devices need to send a message, $\Theta(n \log_n(N))$ time
is necessary and sufficient in the deterministic \strong{} model.

In multi-hop radio networks, \LeaderElection{} and its related problems (e.g., broadcasting and gossiping) have been studied
extensively, where the bounds typically depend on both the
diameter and size of the network, whether it is directed, and whether randomization
and collision detection are available.  See, e.g.,~\cite{bar1991efficient,bar1992time,chlebus2012electing,clementi2003distributed,czumaj2003broadcasting,KushilevitzM98,alon1991lower,KowalskiP05,KowalskiP09,ghaffari2013near,CzumajD16,CzumajD17,Censor-HillelHHZ17}.
Schneider and Watterhofer~\cite{SchneiderW10} investigated the use of collision detection
in multihop radio networks when solving archetypal problems such as MIS, $(\Delta+1)$-coloring, and broadcast.
Their results showed that the value of collision detection depends on the problem being solved.

Cornejo and Kuhn~\cite{cornejo2010deploying} introduced the {\em beeping} model, where no messages are sent; the only signals
are $\lambda_N$ and $\lambda_S$: noise and silence.
The complexity of \ApproximateCounting{} was studied in~\cite{Brandes2016}
and the ``state complexity'' of \LeaderElection{} was studied in~\cite{gilbert2015computational}.

In adversarial settings a {\em jammer} can interfere with communication.
See~\cite{kutylowski2003adversary,daum2012leader} for leader election protocols resilient to jamming.
In a {\em resource-competitive} protocol~\cite{BenderFMSDGPY15}
the energy cost of the devices is some function of the energy cost of the jammer.
See~\cite{BenderFGY16} for resource-competitive contention resolution,
and~\cite{GilbertKPPSY14,KingSY11} for resource-competitive point-to-point communication and broadcast protocols.

\subsection{Organization and Technical Overview}

To establish the two sharp exponential separations we need 8 distinct upper and lower bounds.
The $O(\log N)$ upper bound on deterministic \nocd{} \LeaderElection{} is trivial
and the matching lower bound in \recv{} is provided in~\cite{JurdzinskiKZ02c}.
The $O(\log(\log^* n))$ upper bound from~\cite{BenderKPY16}
on randomized \LeaderElection\ and \ApproximateCounting{} works only in \strong.
This paper contains proofs of all remaining upper and lower bounds.
In addition, we offer  a simpler proof of the $\Omega(\log N)$ lower bound in deterministic \recv, and provide
 an $O(\alpha(N))$ energy protocol for \Census\ in deterministic \nocd\ when $n=\Theta(N)$.


\paragraph{Lower Bounds.}
In Section~\ref{sct:detlb} we begin with a surprisingly simple proof that protocols solving any
non-trivial problem in the deterministic \strong{} model require $\Omega(\log\log N)$ energy
if the devices are adaptive and $\Omega(\log N)$ if they are non-adaptive.
It turns out that \recv{} algorithms are essentially forced to be non-adaptive,
so this yields~$\Omega(\log N)$ lower bounds for deterministic \LeaderElection{} in \recv.
The $\Omega(\log\log N)$ lower bound combines a decision tree representation of the algorithm
with the encoding argument that Katona and \Szemeredi~\cite{KatonaS67} used to solve the
biclique covering problem of \Erdos, \Renyi, and \Sos~\cite{ErdosRS66}.

In Section~\ref{app-sect:randLB} we prove the $\Omega(\log^* n)$ and
$\Omega(\log(\log^* n))$ lower bounds on
randomized \ApproximateCounting\ and \LeaderElection.
These lower bounds begin by embedding any algorithm into an infinite
\emph{universal} DAG that is basically a decision tree with some reconvergent paths.
The proof is information theoretic. There are only two methods
for devices in \strong\ and \recv{} to learn new information. The first method is via direct communication (in which one device successfully transmits a message, and some subset of devices listen); the second method is via inference (in which transmitting or listening devices detect a collision or silence,
which informs their future decisions).
The information theoretic capacity of the first method is essentially unbounded whereas the second method
is bounded by $1$-bit per unit energy in \strong{} and usually less in \recv.
We show that any algorithm with a reasonable time bound can be
forced to learn an approximation of $n$ via the information theoretically well behaved second method.

\paragraph{Upper Bounds.}
In Sections~\ref{sct:detalg} and~\ref{sec:detle_dense} we present all deterministic upper bounds:
an $O(\log \log N)$ energy protocol for \Census,
and an $O(\alpha(N))$ energy protocol for \emph{dense} \Census, when $n=\Theta(N)$.
Notice that a protocol for \Census\ also solves \LeaderElection.
The first protocol combines \emph{van Emde Boas}-like recursion with
a technique that lets a group of devices function \emph{as one device} and thereby share energy costs.

In Section~\ref{app-sect:randUB} we present upper bounds
on randomized \LeaderElection\ and \ApproximateCounting.
When time is not too constrained,
the \weak{} and \recv{} protocols have energy complexity $O(\log^* n)$ and $O(\log(\log^* n))$.
Our protocols naturally adapt to any time bound that is $\Omega(\log^2 n)$, where the energy complexity
gradually increases as we approach this lower limit.  See Table~\ref{table:time-energy-intro}.
These protocols are randomized, and so do not assume distinct IDs; nonetheless, they
use the deterministic $\alpha(N)$ dense \Census{} algorithm of Section~\ref{sec:detle_dense}.


\section{Deterministic Lower Bounds}\label{sct:detlb}

In this section we prove deterministic lower bounds for the \SuccessfulCommunication{} problem, 
which immediately lead to the same lower bounds for \LeaderElection.
The goal of \SuccessfulCommunication{} is to have {\em some} time slot where exactly one device transmits 
a message while at least one other device listens to the channel.
Once a successful communication occurs, the algorithm is terminated on all devices.
Throughout the section, we focus on the special case of $n=2$. Each device knows that $n=2$, but not the ID of the other device.
In this case, the \strong{} and \weak{} models are the same, and the \recv\ and  \nocd\ models are the same.
Theorem~\ref{thm:detlbna} has been proved in~\cite{JurdzinskiKZ02c}, in this section we offer a simpler proof.

\begin{figure}[!ht]
  \centering
  \subfloat[]{
    \includegraphics[width=0.3\linewidth]{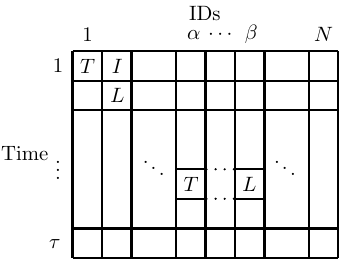}\label{subfig:nonadapt}
  }\
  \subfloat[]{
    \includegraphics[width=0.55\linewidth]{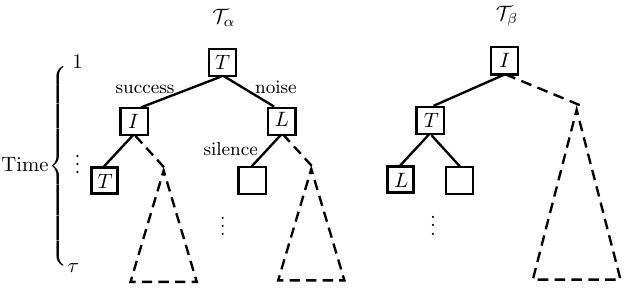}\label{subfig:adapt}
  }
  \caption{(a) Table for an non-adaptive algorithm. 
  (b) Two binary decision trees $\mathcal{T}_\alpha$ and $\mathcal{T}_\beta$ for an adaptive algorithm.
  }
  \end{figure}

\begin{theorem}\label{thm:detlbna}
The deterministic energy complexity of \LeaderElection{} is $\Omega(\log N)$ in \nocd\ and \recv, even when $n=2$.
\end{theorem}
\begin{proof}
For the case of $n=2$ in \nocd\ and \recv, the two devices receive no feedback from the channel until the first successful communication  occurs.
Thus, to prove the theorem, it suffices to show that the energy cost of any {\em non-adaptive} deterministic algorithm $\mathcal{A}$ for \SuccessfulCommunication{} is $\Omega(\log N)$. In a non-adaptive algorithm, the sequence of actions taken by a device is solely a function of its ID, not the information
it receives from the channel.

Let $\tau=\tau(N)$ be the running time of $\mathcal{A}$.
This algorithm can be encoded by a table in the set $\{T,L,I\}^{\tau\times N}$; see Figure~\ref{subfig:nonadapt}.
The $(j,i)$-entry of the table is the action $T$ (\send), $L$ (\listen), or $I$ (\idle) taken by the device of ID $i$ at time $j$.
Let $E_i$ be the energy cost of device of ID $i$, which equals the number of $T$ and $L$ entries in the $i$th column.

We now prove that $\max_i E_i \ge \log N$. The proof is inspired by Katona and \Szemeredi's \cite{KatonaS67}
lower bound of the biclique covering problem.
Encode the $i$th column by a binary string of length $\tau$ by replacing $T$ with $0$, $L$ with $1$, and $I$ with either $0$ or $1$.
There are $2^{\tau-E_i}$ possible encodings for column $i$.
To solve \SuccessfulCommunication, the two devices in the network must successfully communicate at some time slot.
Thus, for any two distinct IDs $\{\alpha, \beta\}$,
there must be a row $r$ such that the $(r,\alpha)$- and $(r,\beta)$-entries of the table contain one $T$ and one $L$.
Therefore, no binary string is an eligible encoding of two distinct columns.
Since there are $2^\tau$ possible encodings, we have:
  \[
  \sum_{i=1}^N 2^{\tau-E_i}\leq 2^\tau, \;\; \mbox{ which implies } \;\; \sum_{i=1}^N \frac{1}{2^{E_i}}\leq 1.
  \]
  This implies $\max_i E_i \geq \log N$.
  Moreover, the convexity of $f(x)=2^{-x}$ implies $N \cdot 2^{-\sum_{i=1}^N {E_i / N}} \leq 1$, 
  and so $\sum_i E_i\geq N\log N$. Thus, even on average the energy cost of $\mathcal{A}$ is $\Omega(\log N)$.
\end{proof}

\begin{theorem}\label{thm:det_lb}
The deterministic energy complexity of \LeaderElection{} is $\Omega(\log\log N)$ in \strong\ and \weak, even when $n=2$.
\end{theorem}

\begin{proof}
It suffices to show that the energy cost of any deterministic algorithm for \SuccessfulCommunication{} is $\Omega(\log \log N)$.
Suppose we have an algorithm  $\mathcal{A}$ for \SuccessfulCommunication{} running in $\tau$ time when $n=2$.
We represent the behavior of the algorithm on the device with ID $i$ as a binary decision tree $\mathcal{T}_i$.
Each node in $\mathcal{T}_i$ is labeled by $T$ (\send), $L$ (\listen), or $I$ (\idle).
An $I$-node has one left child and no right child; a $T$-node has
two children, a left one indicating collision-free transmission, and a right one indicating a collision;
an $L$-node has two children, a left one indicating silence, and a right one indicating that the device receives a message. Notice that the algorithm terminates once a device reaches the right child of an $L$-node in the decision tree.

The left-right ordering of children is meaningless but {\em essential} to making the following argument work. Suppose that we run
$\mathcal{A}$  on two devices with IDs $\alpha$ and $\beta$. Let $t$ be the first time a   successful
communication occurs.  We claim that the paths in $\mathcal{T}_\alpha$ and $\mathcal{T}_\beta$ corresponding to the execution of $\mathcal{A}$ have exactly the  same sequence of $t-1$ left turns and right turns.  At any time slot before $t$ the possible actions performed by $\{\alpha,\beta\}$
are $\{I,I\}, \{I,T\}, \{I,L\}, \{L,L\}, \{T,T\}$.  In all cases, both $\alpha$ and $\beta$ branch to the left child, except
for $\{T,T\}$, where they both branch to the right child.  At time $t$ the actions of the two devices are $\{T,L\}$, and it is only here that they
branch in different directions.  See Figure~\ref{subfig:adapt}.

We extend each $\mathcal{T}_i$ to a full binary tree of depth $\tau$ by adding dummy nodes.
The number of nodes in a full binary tree of depth $\tau$ is $2^\tau-1$, and so we encode $\mathcal{T}_i$ by a binary string of length $2^\tau-1$
by listing the nodes in any fixed order (e.g.,  pre-order traversal),
mapping each $T$-node to $0$, $L$-node to 1, and each $I$-node or dummy node to either 0 or 1.
For any two distinct IDs $\{\alpha,\beta\}$, there must be a position in the full binary tree such that the corresponding two nodes in
 $\mathcal{T}_\alpha$ and $\mathcal{T}_\beta$ are one $T$-node and one $L$-node. Therefore, no binary string
is an eligible encoding of $\mathcal{T}_\alpha$ and $\mathcal{T}_\beta$.
If a device with ID $i$ spends energy $E_i$, then the number of $T$-nodes and $L$-nodes in $\mathcal{T}_i$ is at most $2^{E_i}-1$,
and so $\mathcal{T}_i$ has at most $2^{(2^\tau-1) - (2^{E_i}-1)}$ possible encodings.
Thus,
\[
\sum_{i=1}^N 2^{(2^\tau-1)-(2^{E_i}-1)} \leq 2^{2^\tau-1}, \;\; \mbox{ which implies } \;\; \sum_{i=1}^N \frac{1}{2^{2^{E_i}-1}}\leq 1.
\]
This implies $\max_i E_i \ge \log(\log N + 1)$.  Moreover,
the convexity of $f(x)=2^{-(2^x - 1)}$ implies $N \cdot 2^{-(2^{\sum_{i=1}^N {E_i / N}} - 1)} \leq 1$, and so $\sum_i E_i\geq N\log(\log N+1)$. 
Thus, even on average the energy cost of $\mathcal{A}$ is $\Omega(\log \log N)$.
\end{proof}

\section{Randomized Lower Bounds}\label{app-sect:randLB}

In this section we prove energy lower bounds of randomized algorithms for \ApproximateCounting{}.
Since \nocd{} is strictly weaker than \weak, the $\Omega(\log^\ast n)$ lower bound also applies to \nocd.
Similarly, the $\Omega(\log (\log^\ast n))$ lower bound for \strong{} also applies to \recv.

\begin{theorem}\label{thm:randLB-1}
The energy cost of any polynomial time \ApproximateCounting{} algorithm with
failure probability $1/n$ is $\Omega(\log^\ast n)$ in the \weak\ and \nocd\ models.
\end{theorem}

\begin{theorem}\label{thm:randLB-2}
The energy cost of any polynomial time \ApproximateCounting{} algorithm with
failure probability $1/n$ is $\Omega(\log (\log^\ast n))$ in the \strong\ and \recv\ models.
\end{theorem}

In Section~\ref{sec:decision-tree} we introduce the randomized decision tree, which is the foundation of our lower bound proofs.
In Section~\ref{sec:weakcdlb}, we prove Theorem~\ref{thm:randLB-1}.
In Section~\ref{sec:strongcdlb}, we prove Theorem~\ref{thm:randLB-2}.
In Section~\ref{sec:otherlb}, we demonstrate that our lower bounds proofs can be adapted to other problems such as \LeaderElection, and prove the impossibility of {\em loneliness detection} (i.e., distinguishing between $n=1$ and $n>1$) in randomized $\nocd$.

\subsection{Randomized Decision Tree}
\label{sec:decision-tree}
The process of a device $s$ interacting with the network at time slot $t$ has two phases. During the first phase (action performing phase), $s$ decides on its action, and if this action is to \send, then $s$ chooses a message $m \in \mathcal{M}$ and transmits $m$. During the second phase (message receiving phase), if $s$ chose to \listen{} or \send{} during the first phase, then  $s$ may receive a feedback from the channel which depends on the transmissions occurring at this time slot and the collision detection model.
The phases partition the time into {\em layers}. We write layer $t$ to denote the time right before the first phase of time slot $t$, and layer $t + 0.5$ to denote the time right before the second phase of time slot $t$. The choice of the message space $\mathcal{M}$ is irrelevant to our lower bound proof. The cardinality of $\mathcal{M}$ may be finite or infinite.

For a device $s$, the {\em state} of $s$ at layer $t$ includes the ordered  list of actions taken by $s$ and feedback
received from the channel until layer $t$.
There is only one possible state in layer $1$, which is the common {\em initial state} of all devices before the execution of an algorithm.

Our lower bounds are proved using a {\em single} decision tree $\mathcal{T}$, which has unbounded branching factor if $|\mathcal{M}|$ is unbounded. A special directed acyclic graph (DAG) $\mathcal{G}$ is defined to capture the behaviour of {\em any} randomized algorithm, and then the decision tree $\mathcal{T}$ is constructed by ``short-cutting'' some paths in $\mathcal{G}$.

\paragraph{DAG $\mathcal{G}$.}  The nodes in $\mathcal{G}$ represent all possible states of a device during the execution of any algorithm. Similarly, the arcs represent all legal transitions between states during the execution of any algorithm. Therefore, each arc connects only nodes in adjacent layers, and the root of $\mathcal{G}$ is the initial state.

Let $t \in \mathds{Z}^+$. A transition from a state $u$ in layer $t$ to a state $v$ in layer $t + 0.5$ corresponds to one of the possible $|\mathcal{M}|+2$ actions that can be performed in the first phase of time slot $t$ (i.e., \send\ $m$ for some $m \in \mathcal{M}$, \listen, or \idle). The transitions from a state $u$ in layer $t + 0.5$ to a state $v$ in layer $t+1$ are more involved.
Based on the action performed in the first phase of time slot $t$ that leads to the state $u$, there are three cases:
\begin{itemize}
\item Case: the action is \idle. The state $u$ has one outgoing arc corresponding to doing nothing.
\item Case: the action is \listen. The state $u$ has $|\mathcal{M}|+2$ outgoing
arcs in \strong{}, or $|\mathcal{M}|+1$ in  \weak, corresponding to all possible channel feedbacks that can be heard.
\item Case: the action is \send. The state $u$ has two outgoing arcs. The first (resp., second) outgoing arc corresponds to the message transmission succeeding (resp., failing). If a failure took place, then no other device knows which message was sent by the device, and so the content of this message is irrelevant. Thus, all states $u$ in layer $t + 0.5$ that correspond to the action \send{} and share the same parent have the same child node in layer $t+1$ corresponding to a failure in transmitting the message.
The arcs corresponding to failed transmissions are what makes $\mathcal{G}$ a DAG rather than a tree.
\end{itemize}

\paragraph{Embedding an Algorithm.} Any algorithm $\mathcal{A}$ can be embedded into $\mathcal{G}$, as follows.
First of all, appropriate states, depending on $\mathcal{A}$, are designated as {\em terminal states}. Without loss of generality, we require that any terminal state must be in layer $t$ for some $t \in \mathds{Z}^+$.
Each terminal state is labelled with a specific output for the problem at hand. A device entering a terminal state $u$ terminates with the output associated with the state $u$. Any randomized algorithm is completely described by designating the terminal states together with  their outputs, and specifying the transition probabilities from states in layer $t$ to states in layer $t + 0.5$ for all $t \in \mathds{Z}^+$.

\paragraph{Randomized Decision Tree $\mathcal{T}$.}  The  tree $\mathcal{T}$ is derived from $\mathcal{G}$ as follows. The set of nodes of $\mathcal{T}$  is the set of nodes in $\mathcal{G}$  that are in layer $t$  for some  $t \in \mathds{Z}^+$. For any two states $u$ in  layer $t \in \mathds{Z}^+$ and $v$ in  layer $t+1$ that are linked by a directed path in  $\mathcal{G}$, there is a  transition from $u$ to $v$ in  $\mathcal{T}$.
It is straightforward to see that $\mathcal{T}$ is a rooted tree. See Figure~\ref{fig-tree} for an illustration of both  $\mathcal{G}$ and $\mathcal{T}$ in the \strong{} model with $\mathcal{M}=\{m_1, \ldots, m_k\}$. Notice that in the \strong{} model, a device transmitting a message $m_i$ to the channel at a time slot must not hear $\lambda_S$ in the same time slot. If the transmission is successful, it hears the message $m_i$; otherwise it hears $\lambda_N$.

For a state $u$ in layer $t \in \mathds{Z}^+$, and for an action $x \in \{\idle,\listen,\send\}$, we write $\pp{u}{x}$ to denote the probability that a device in state $u$ performs action $x$ in the first phase of time slot $t$.

\begin{figure}[h]
\begin{center}
\includegraphics[width=0.9\textwidth]{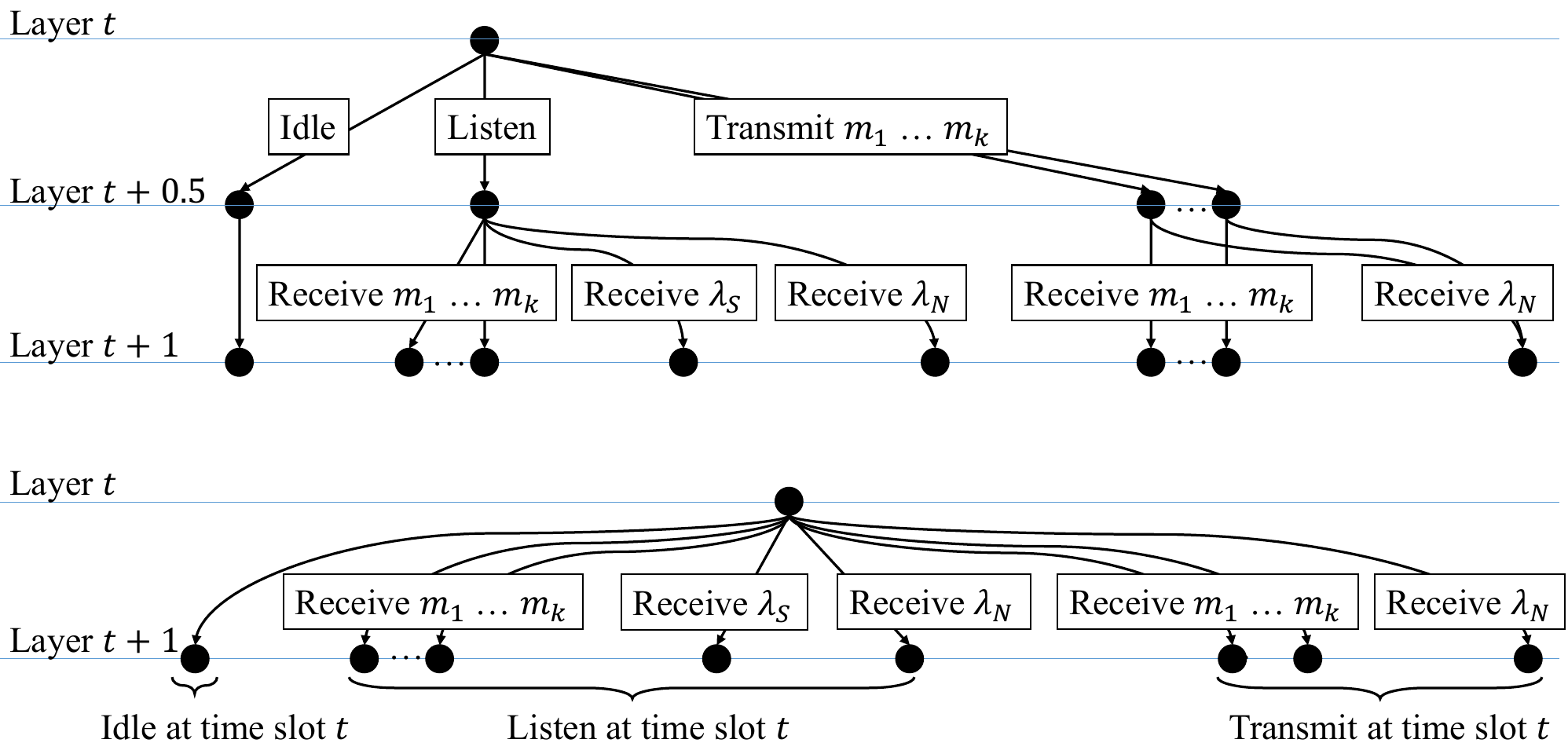}
\end{center}
    \caption{Upper: a portion of  $\mathcal{G}$. Lower: the corresponding portion in $\mathcal{T}$.
    } \label{fig-tree}
\end{figure}

\paragraph{Time and Energy Complexity.}
An execution of an algorithm for a device is completely described by a directed path $P = (u_1, u_2, \ldots, u_k)$ in $\mathcal{T}$ such that $u_t$ is in time slot $t$ for each $1\leq t \leq k$, and $u_k$ is the only terminal state in $P$. The runtime of the device is $k$. The amount of energy the device spends is the number of transitions corresponding to \listen{} or \send{} in $P$.
The time (resp., energy) of an execution of an algorithm is the maximum time (resp., energy) spent by any device.

\subsection{Lower Bound in the \weak\ Model}
\label{sec:weakcdlb}
In this section we prove Theorem~\ref{thm:randLB-1}.
Let  $\mathcal{A}$ be any $T(n)$ time   algorithm for \ApproximateCounting{} in  \weak{}  with failure probability at most $1/n$.
We  show that the energy cost of $\mathcal{A}$  is $\Omega(\log^\ast n)$.


\paragraph{Overview.}
The high level idea of the proof is as follows. We will carefully select a sequence of network sizes $\{n_i\}$ with {\em checkpoints} $\{d_i\}$ such that $d_i < n_i  < d_{i+1}$ and $T(n_i) < d_{i+1}$.
There are two main components in the proof. The first component is to  demonstrate that, with probability  $1-1/\poly(n_i)$, no message is successfully transmitted before time $d_i$ when running $\mathcal{A}$ on $n_i$ devices, i.e., every transmission ends in a collision. This limits the amount of information that could be learned from a device.
The second component is to prove that, for $j>i$, in order for a device $s$ to learn enough information to distinguish between $n_i$ and $n_j$ within $T(n_i) < d_{i+1}$ time slots, the device $s$ must use at least one unit of energy within time interval $[d_i, d_{i+1}-1]$. The intuition is briefly explained as follows. Given that $n \in \{n_i, n_j\}$, with high probability, every transmission ends in a collision before time $d_i$, and so $s$ has not yet obtained enough information to  distinguish between $n_i$ and $n_j$ by the time $d_i - 1$. The only way $s$ can gain information is to use energy, i.e., \listen{} or \send. It is required that $s$ terminates by time $T(n_i)$ if the total number of devices is $n_i$, and so $s$ must use at least one unit of energy within time interval $[d_i, T(n_i)] \subseteq [d_i, d_{i+1}-1]$.

\paragraph{Truncated Decision Tree.} The {\em no-communication tree} $\Tideal$ is defined as the subtree of $\mathcal{T}$ induced by the set of all states $u$ such that no transition in the path from the root to $u$ corresponds to receiving a message in $\mathcal{M}$. In other words, $\Tideal$ contains exactly the states whose execution history contains no successful communication.
Notice that in \weak{} each state in $\Tideal$ has exactly three children, and the three children correspond to the following three pairs of action performed and channel feedback received:  $(\send,\lambda_S)$, $(\listen,\lambda_S)$, and $(\idle, \text{N/A})$.

For each state $u$ at layer $t$ of the tree $\Tideal$, we define the {\em probability estimate} $p_u$ inductively as follows. If $u$ is the root, $p_u = 1$; otherwise $p_u = p_v \cdot \pp{v}{x}$, where $v$ is the parent of $u$, and $x$ is the action  performed at time slot $t-1$ that leads to the state $u$. Recall that $\pp{v}{x}$ is defined as  the probability for a device in $v$ (which is a state in layer $t-1$) to perform $x$ at time slot $t-1$.
Intuitively, if no message is successfully sent in an execution of $\mathcal{A}$, the proportion of devices entering $u$ is
well concentrated around $p_u$, given that $p_u$ is high enough. See Figure~\ref{fig-nocomm} for an illustration of
no-communication tree $\Tideal$ and probability estimates in the \weak{} model.

\begin{figure}[h]
\begin{center}
\includegraphics[width=0.9\textwidth]{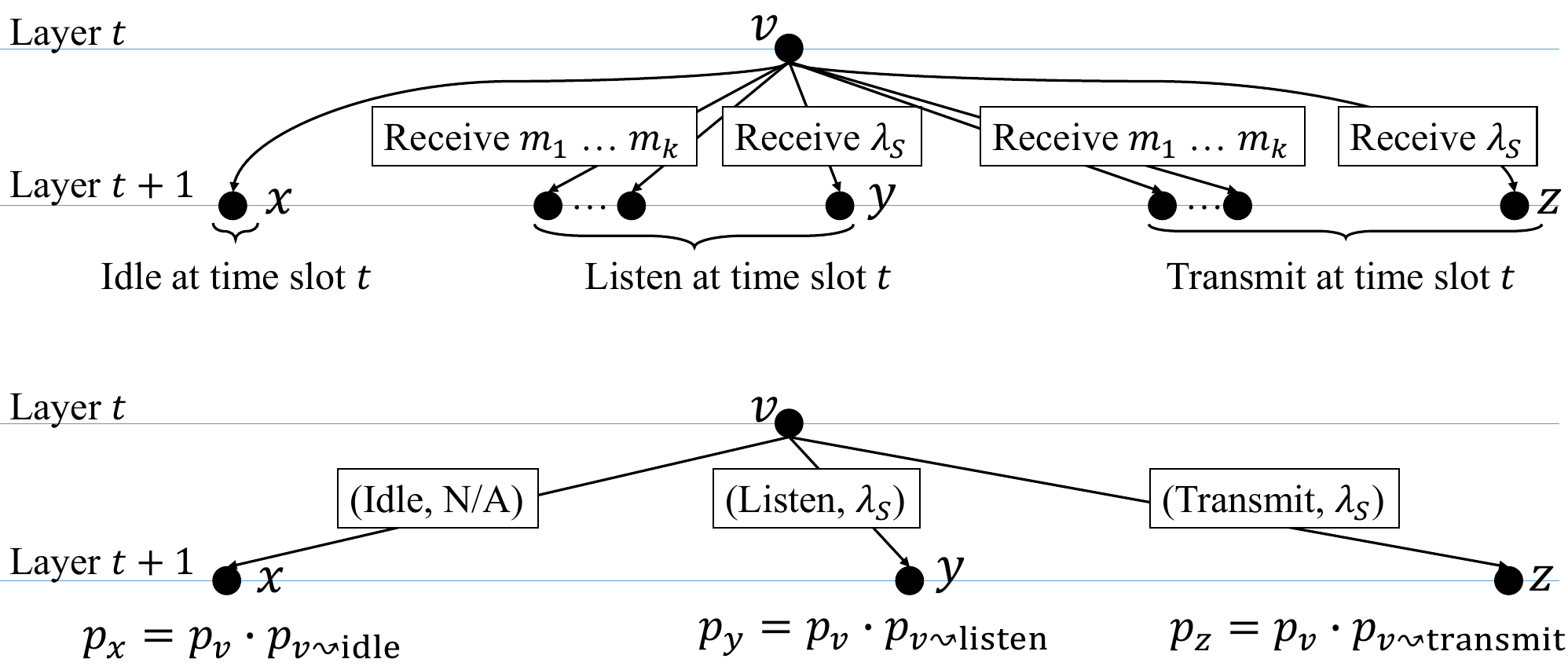}
\end{center}
    \caption{Upper: a portion of the tree $\mathcal{T}$ in \weak. Lower: the corresponding portion in the no-communication tree $\Tideal$.
    } \label{fig-nocomm}
\end{figure}

Given the runtime constraint $T(n)$ for $\mathcal{A}$, we select an infinite sequence of {\em checkpoints} as follows: $d_1$ is a sufficiently large constant to meet the requirements in the subsequent analysis; for each $i > 1$,  $d_i$ is any number satisfying the two criteria (i) $d_i \geq 2^{{2^{2^{2^{d_{i-1}}}}}}$ and (ii) $d_i \geq T(n')+1$ for all $2^{2^{d_{i-1}}} < n' < 2^{2^{2^{2^{d_{i-1}}}}}$. For example, if $T(n)$ is a non-decreasing function and $T(n) \geq n$, then we can simply set $d_i = T\left(2^{{2^{2^{2^{d_{i-1}}}}}}\right) + 1$.\footnote{It might be possible to reduce the height of the power tower. However,  having {\em any} constant height is sufficient to prove the desired lower bound, so we do not to optimize the constant.}

\begin{lemma}\label{lem-ni-exists-app}
For each index $i$, there exists a number $n_i$ such that $2^{2^{d_i}} < n_i < 2^{2^{2^{2^{d_i}}}}$ and for each state $u \in \Tideal$ at layer at most $d_i$, either $p_u \leq n_i^{-10}$ or $p_u \geq n_i^{-1/10}$.
\end{lemma}
\begin{proof}
Define $m_1 = 2^{2^{d_i}}+1$;  for $1 < k \leq 3^{d_i}$, define $m_k = m_{k-1}^{100} + 1$.
It is straightforward to see that  $2^{2^{d_i}} < m_1 < m_2 <  \ldots < m_{3^{d_i}} < 2^{2^{2^{2^{d_i}}}}$, so long as $d_i$ is greater than some universal constant.
For each state $u \in \Tideal$ at layer at most $d_i$, there exists at most one $m_k$ with ${m_k}^{-10} < p_u <  {m_k}^{-1/10}$.
Recall that $\Tideal$ has branching factor $3$, and hence the number of states up to layer $d_i$ is less than $3^{d_i}$.
By the pigeonhole principle, among the $3^{d_i}$ distinct integers $m_1, m_2, \ldots, m_{3^{d_i}}$, there exists one integer $n_i$ such that, for each state $u \in \Tideal$ at layer at most $d_i$, either $p_u \leq n_i^{-10}$ or $p_u \geq n_i^{-1/10}$.
\end{proof}

For each index $i$, the parameter $n_i$ is chosen to meet the statement of Lemma~\ref{lem-ni-exists-app}.
Recall that the goal of  $\mathcal{A}$ is to calculate an estimate $\tilde{n}$ that is within a multiplicative factor $c$ of $n$, where $c > 1$ is some constant.
We select the first checkpoint $d_1$  to be a large enough constant such that $c \cdot n_i < n_{i+1}$ for all $i$. We define $\mathcal{T}_{i}$ as the subtree of $\Tideal$ that consists of all states $u$ up to layer $d_i$ such that $p_u \geq n_i^{-1/10}$. Notice that $\mathcal{T}_{i} \subseteq \mathcal{T}_{i+1}$ for all $i$.

Consider an  execution of  $\mathcal{A}$ on $n_i$ devices.
Let $t \in [1, d_i]$, and denote $\mathcal{P}_{i}^{(t)}$ as the event that, for each state $u$ in layer $t$ of the decision tree $\mathcal{T}$, the number of devices entering $u$ is within the range
$n_i \cdot p_u \pm (t-1) \cdot n_i^{0.6}$ if $u$ is in layer $t$ of  $\mathcal{T}_{i}$,
and is $0$ if $u \notin \mathcal{T}_{i}$. We write $\mathcal{P}_{i} = \bigcap_{t = 1}^{d_i} \mathcal{P}_{i}^{(t)}$.
Notice that $\mathcal{P}_{i}^{(1)}$ holds with probability 1.

\begin{lemma}\label{lem-aux-1}
Let $t < d_i$ and $x \in \{\send,\listen,\idle\}$. Let $v$ be a state in  layer $t$ of $\mathcal{T}_{i}$ such that the number of devices entering $v$ is within $n_i \cdot p_v \pm (t-1) \cdot n_i^{0.6}$.
Define $m$ as the number of devices that are in state $v$ and perform action $x$ at time $t$.
The following holds with probability $1 - O(n_i^{-9})$.
If $p_v \cdot \pp{v}{x} \geq n_i^{-1/10}$, then $m$ is within $n_i \cdot p_v \cdot \pp{v}{x} \pm t \cdot n_i^{0.6}$;
 if $p_v \cdot \pp{v}{x} \leq n_i^{-10}$, then $m = 0$.
\end{lemma}
\begin{proof}
According to our choice of $n_i$, for each state $u \in \Tideal$ at layer at most $d_i$, either $p_u \leq n_i^{-10}$ or $p_u \geq n_i^{-1/10}$. Since $p_v \cdot \pp{v}{x} = p_u$ for some $u \in \Tideal$ at layer at most $d_i$, either (i) $n_i^{-10} \geq   p_v \cdot \pp{v}{x}$ or (ii)  $p_v \cdot \pp{v}{x} \geq n_i^{-1/10}$ is true.

First, consider the case $n_i^{-10} \geq   p_v \cdot \pp{v}{x}$. Notice that $\pp{v}{x} \leq n_i^{-10} / p_v \leq n_i^{-9.9}$, and recall $t \leq d_i < \log \log n_i$. We upper bound the expected value of $m$ as follows.
\begin{align*}
\Expect[m] &\leq \left(n_i \cdot p_v + (t-1) \cdot n_i^{0.6}\right) \cdot \pp{v}{x} \\
 &\leq     n_i^{-9} + (t-1) \cdot n_i^{0.6}\cdot  \pp{v}{x} \\
 &\leq    n_i^{-9} + (t-1) \cdot n_i^{-9.3} \\
 &< 2 n_i^{-9}.
\end{align*}
By Markov's inequality, $m = 0$ with probability at least $1 - 2 n_i^{-9}$, as desired.

Next, consider the case $p_v \cdot \pp{v}{x} \geq n_i^{-1/10}$. The value of $\Expect[m]$ is within $\pp{v}{x} \cdot (n_i \cdot p_v \pm (t-1) \cdot n_i^{0.6})$, which is within the range
$n_i \cdot p_v \cdot \pp{v}{x} \pm (t-1) \cdot n_i^{0.6}$. Let $\delta = n_i^{-0.4}/2$. A simple calculation shows that $\delta \cdot \Expect[m] < n_i^{0.6}$ and $\Expect[m] > n_i^{0.9}/2$. Notice that each device in the state $v$ decides which action to perform next independently.
By a Chernoff bound, the probability that $m$ is within $1 \pm \delta$ factor of $\Expect[m]$ is
at least $1 - 2\exp(-\delta^2 \cdot \Expect[m] /3) \geq 1 - 2\exp(-(n_i^{-0.4}/2)^2 \cdot (n_i^{0.9}/2) / 3) > 1 - O(n_i^{-9})$.
Therefore, with such probability, $m$ is in the range $\Expect[m] (1\pm \delta)$, which is within $n_i \cdot p_v \cdot \pp{v}{x} \pm t \cdot n_i^{0.6}$, as desired.
\end{proof}

\begin{lemma}\label{lem-find-ni}
For an  execution of  $\mathcal{A}$ on $n_i$ devices, $\mathcal{P}_{i}$ holds with probability at least $1 -  n_i^{-7}$.
\end{lemma}
\begin{proof}
For the base case, $\mathcal{P}_{i}^{(1)}$ holds with probability 1.
For each $1 < t \leq d_i$, we will show that $\Prob\left[ \mathcal{P}_{i}^{(t)} \ \middle| \ \mathcal{P}_{i}^{(t-1)} \right] = 1 -  O(n_i^{-8})$.
Therefore, by a union bound on all $t \in \{1, \ldots, d_i\}$, we have:
$$
\Prob[\mathcal{P}_{i}] = \Prob\left[\bigcap_{t = 1}^{d_i} \mathcal{P}_{i}^{(t)}\right] \geq 1 - d_i \cdot O(n_i^{-8}) \geq 1 - O(n_i^{-8} \log \log n_i)
\geq 1 -  n_i^{-7}.$$

Let $1 < t \leq d_i$. Suppose that $\mathcal{P}_{i}^{(t-1)}$ holds. This implies that, for each state $v$ in  layer $t-1$ of $\mathcal{T}_{i}$, the number of devices entering $v$ is within $n_i \cdot p_v \pm (t-2) \cdot n_i^{0.6}$.
The statement of Lemma~\ref{lem-aux-1} holds for at most $3^{t-1}$ choices of states $v$ in layer  $t-1$ of $\mathcal{T}_{i}$, and all 3 choices of $x \in \{\send,\listen,\idle\}$, with probability at least
$$1 - 3 \cdot 3^{t-1} \cdot O(n_i^{-9}) \geq 1 -  O(n_i^{-9} \poly \log n_i) \geq 1 -  O(n_i^{-8}).$$
In particular, this implies that, with probability $1 -  O(n_i^{-8})$, at time $t-1$, the number of devices transmitting  is either 0 or greater than 1, which implies that no message is successfully sent.
Therefore, at layer $t$, all devices are confined in states within $\Tideal$.
Let $u$ be the child of $v$ in $\Tideal$ such that the arc $(v,u)$ corresponds to action $x$.
Due to our choices of $n_i$ and $\mathcal{T}_{i}$, if $u \in \mathcal{T}_{i}$, then $p_u = p_v \cdot \pp{v}{x} \geq n_i^{-1/10}$; if $u \notin \mathcal{T}_{i}$, then $p_u = p_v \cdot \pp{v}{x} \leq n_i^{-10}$.
Therefore, in view of the statement of Lemma~\ref{lem-aux-1}, $\mathcal{P}_{i}^{(t)}$ holds with probability at least $1 -  O(n_i^{-8})$.
\end{proof}

\begin{lemma}\label{lem-term-state}
The no-communication tree $\Tideal$ has no terminal state $u$ with $p_u \neq 0$.
\end{lemma}
\begin{proof}
Suppose that $u \in \Tideal$ is a terminal state with $p_u \neq 0$. Then there exists an index $i$ such that for {\em all} $j \geq i$, $u \in \mathcal{T}_{j}$. Among all $\{n_j\}_{j\geq i}$, the decision of  $u$ is a correct estimate of at most one $n_j$. Therefore, the adversary can choose one network size $n_{j'}$ from $\{n_j\}_{j\geq i}$ such that when $\mathcal{A}$ is executed on $n_{j'}$ devices, any device entering $u$ gives a wrong estimate of $n_{j'}$. By Lemma~\ref{lem-find-ni}, with probability $1 -  n_{j'}^{-7} > 1/{n_{j'}}$, there is a device entering $u$, and hence the algorithm fails with probability higher than $1/{n_{j'}}$, a contradiction.
\end{proof}

In the following lemma, we show how to force energy expenditure of devices.

\begin{lemma}\label{lem:force-energy}
Define  $p_{\operatorname{idle}}^{(i)}$ as the maximum probability for a device $s$ that is in a state $u$ in layer $d_i$ of $\mathcal{T}_{i}$ to be \idle{} throughout the time interval $[d_{i}, d_{i+1}-1]$, where the maximum ranges over all states in layer $d_i$ of $\mathcal{T}_{i}$. Then $p_{\operatorname{idle}}^{(i)} < 2^{-d_i}$.
\end{lemma}
\begin{proof}
For a device $s$ to terminate within the time constraint $T(n_i)$, the device $s$ must {\em leave} the tree $\Tideal$ by time $T(n_i) < d_{i+1}$ due to Lemma~\ref{lem-term-state}. Suppose that the device $s$ is currently in a state $u$ in layer $d_i$ of $\mathcal{T}_{i} \subseteq \Tideal$. In order to leave the tree $\Tideal$  by time $T(n_i)$, the device $s$ must  successfully hear some message $m \in \mathcal{M}$  by time $T(n_i)$.
Since $\mathcal{T}_{i} \subseteq \Tideal$, $s$ has not heard any message by time $d_i - 1$, and so at least one unit of energy expenditure in the time interval $[d_{i}, d_{i+1}-1]$ is required.

Recall that if $\mathcal{P}_{i}$ occurs, then all devices are confined in $\mathcal{T}_{i}$ up to layer $d_i$.
If we execute $\mathcal{A}$ on $n_i$ devices, then the probability that the runtime of a device exceeds $T(n_i)$ is at least $\Prob[\mathcal{P}_{i}] \cdot  p_{\operatorname{idle}}^{(i)}$, and so we must have $1/{n_i} \geq \Prob[\mathcal{P}_{i}] \cdot  p_{\operatorname{idle}}^{(i)}$.
By Lemma~\ref{lem-find-ni}, we have $\Prob[\mathcal{P}_{i}] \geq 1 - n_i^{-7} > 1/2$.
Therefore,  $p_{\operatorname{idle}}^{(i)} \leq 1/(n_i \Prob[\mathcal{P}_{i}]) < 2/{n_i} \leq 2 \cdot  2^{-2^{d_i}} < 2^{-d_i}$, as desired.
\end{proof}

We are now in a position to prove the main result of this section.

\begin{lemma}\label{thm-LB-coCD}
For any $i\geq1$, there exists a network size $n$ satisfying $d_i <  n < d_{i+1}$ such that if $\mathcal{A}$ is executed on $n$ devices, for any device $s$, with probability at least $1/2$ the device $s$ spends at least one unit of energy in each of the time intervals $[d_{j}, d_{j+1}-1]$, $1 \leq j \leq i$.
\end{lemma}
\begin{proof}
We select $n = n_i$.
Consider an execution of  $\mathcal{A}$ on $n_i$ devices, and let $s$ be any one of the $n_i$ devices.
Let $j \in \{1, \ldots, i\}$. We claim that, given that $\mathcal{P}_{i}$  holds, with probability $1 - 2 \cdot 2^{-d_j}$ the device $s$ spends at least one unit of energy in  the interval $[d_{j}, d_{j+1}-1]$.
 Then, by a union bound on all $j \in \{1, \ldots, i\}$, the probability that the device $s$ spends at least one unit of energy in each of the intervals $[d_{j}, d_{j+1}-1]$, $1 \leq j \leq i$, is at least $1 - (1-\Prob[\mathcal{P}_{i}]) - 2 \sum_{j=1}^{i} 2^{-d_j}$, which is greater than $1/2$ if $d_1$ is chosen as a sufficiently large constant.

 Next, we prove the above claim. Suppose $\mathcal{P}_{i}$  holds. In view of Lemma~\ref{lem:force-energy}, if $s$ enters a state in layer $d_j$ of $\mathcal{T}_{j}$, then $s$ spends at least one unit of energy in  the time interval $[d_{j}, d_{j+1}-1]$ with probability $1-2^{-d_j}$.
 Thus, all we need to do is to show that the probability that $s$ enters a state in layer $d_j$ of $\mathcal{T}_{j}$ is at least $1 - 2^{-d_j}$.

 Recall that $\mathcal{T}_{j}$ is a subtree of $\mathcal{T}_{i}$. Let $u$ be a state in layer $d_j$ that does not belong to $\mathcal{T}_{j}$. We have $p_u < n_j^{-1/10}$. Since $\mathcal{P}_{i}$  holds, the number of devices entering the state $u$ is at most $n_i \cdot n_j^{-1/10} + (d_j - 1) \cdot n_i^{0.6}$. Since there are at most $3^{d_j}$ states in layer $d_j$ of $\mathcal{T}_{i}$, the proportion of the devices that do not enter a state in layer $d_j$ of $\mathcal{T}_{j}$ is at most
\[
\frac{1}{n_i} \left(n_i \cdot n_j^{-1/10} + (d_j - 1) \cdot n_i^{0.6} \right) \cdot 3^{d_j}
 = \left(n_j^{-1/10} + (d_j - 1) \cdot n_i^{-0.4} \right) \cdot 3^{d_j}
 < 2^{-d_j},
\]
since $n_i \geq n_j \geq 2^{2^{d_j}}$.  
\end{proof}

If it is the case that $T(n) \leq \exp^{(\ell)}(n)$, for some constant $\ell$, where $\exp^{(i)}$ is iterated $i$-fold application of $\exp$,
then it is possible to set the checkpoints such that $\arg \max_i (d_i < n) = \Theta(\log^\ast n)$, and so Lemma~\ref{thm-LB-coCD} implies that the energy cost $\mathcal{A}$ is  $\Omega(\log^\ast n)$. Therefore, we conclude Theorem~\ref{thm:randLB-1} (which is the case of $T(n) = O(\poly (n))$).

\subsection{Lower Bound in the \strong{} Model}
\label{sec:strongcdlb}

In this section we prove Theorem~\ref{thm:randLB-2}.
Let  $\mathcal{A}$ be any $T(n)$ time   algorithm for \ApproximateCounting{} in  \strong{}  with failure probability at most $1/n$.
We  show that the energy cost of $\mathcal{A}$  is $\Omega(\log \log^\ast n)$.

\paragraph{Overview.} Similar to Section~\ref{sec:weakcdlb}, we will construct a sequence of network sizes $\{n_i\}$ with checkpoints $\{d_i\}$ such $d_i < n_i  < d_{i+1}$ and $T(n_i) < d_{i+1}$. Each index $i$ is associated with a truncated decision tree $\mathcal{T}_i$ such that if we execute $\mathcal{A}$ on $n_i$ devices, then the execution history of all devices until time $d_i$ are confined to $\mathcal{T}_i$ with probability  $1-1/\poly(n_i)$.

Suppose that the actual network size $n$ is chosen from the set $S = \{n_1, \ldots, n_k\}$.
The proof in Section~\ref{sec:weakcdlb} says that it costs $\Omega(k)$ energy to estimate $n$, when $n = n_k$.
However,  in the \strong\ model,  the devices are capable of differentiating between \silence{} and \noise, and so they are able to perform a binary search on $S$, which costs only $O(\log k)$ energy to estimate $n$.

The high level idea of our proof of Theorem~\ref{thm:randLB-2} is to demonstrate that this binary search strategy is optimal. We will carefully select a path $P$ in $\Tideal$ reflecting a worst case scenario of the binary search, and we will show that the energy consumption of any device whose execution history follows the path $P$ is $\Omega(\log (\log^\ast n))$.

\paragraph{Basic Setup.} The definitions of the no-communication tree $\Tideal$ and probability estimate $p_u$ are adapted from Section~\ref{sec:weakcdlb}. In the \strong{} model each state in $\Tideal$ has exactly four children, corresponding to all valid combinations of $\{\lambda_S, \lambda_N\}$ and $\{\send, \listen, \idle\}$: $(\send,\lambda_N)$, $(\listen,\lambda_S)$, $(\listen,\lambda_N)$, and $(\idle,\text{N/A})$. Notice that a device transmitting a message never hears \silence{} in the \strong{} model.
The definition of the checkpoints $d_i$ and how we select the network sizes $n_i$ are also the same as in Section~\ref{sec:weakcdlb}.

\paragraph{Truncated Subtrees.}  The subtrees $\{\mathcal{T}_i\}$ are defined differently. For each index $i$, the subtree $\mathcal{T}_i$, along with the sequence $\{m_{i,t}\}_{1 \leq t \leq d_i - 1}$ indicating a likely status (\noise\ or \silence) of the channel at time slot $t$ when $n = n_i$, is constructed layer-by-layer as follows.
\begin{description}
\item[Base Case.] The first layer of $\mathcal{T}_i$ consists of only the initial state.

\item[Inductive Step.]  For each $1 < t \leq d_i$, suppose that layer $t-1$ of $\mathcal{T}_i$ has been defined. If there is at least one state $v$ in layer $t-1$ of $\mathcal{T}_i$ with $p_v \cdot \pp{v}{\send} \geq  {n_i}^{-1/10}$, set $m_{i,t-1} = \lambda_N$; otherwise set $m_{i,t-1} = \lambda_S$.
Let $u$ be a state in layer $t$ that is a child of a state $w$ in $\mathcal{T}_i$. Let $m \in \{\lambda_N, \lambda_S, \text{N/A}\}$ and $x \in \{\send, \listen, \idle\}$ be the channel feedback associated with the arc $(w,u)$. We add $u$ to the layer $t$ of $\mathcal{T}_i$ if  the following two conditions are met: (i) $p_u \geq  {n_i}^{-1/10}$ and (ii) $x = \idle$ or $m = m_{i,t-1}$.
\end{description}
We discuss some properties of $\mathcal{T}_i$.
All states in $\mathcal{T}_i$ are in layers $[1, d_i]$.
Let $w$ be a layer $(t-1)$ state in $\mathcal{T}_i$, and let $u_1$ and $u_2$ be the two children of $v$ corresponding to $(\listen,\lambda_S)$ and $(\listen,\lambda_N)$. Due to the definition of $\mathcal{T}_i$, {\em at most one} of $u_1$ and $u_2$ is in $\mathcal{T}_i$,
and so each state in $\mathcal{T}_i$ has at most three children.
We do \emph{not} have $\mathcal{T}_1 \subseteq \mathcal{T}_2 \subseteq \cdots$ in general.

We define the event $\mathcal{P}_{i}$ in the same way as in Section~\ref{sec:weakcdlb}, but using
the new definition of $\mathcal{T}_i$. We have the following lemma, whose proof is essentially the same as that of Lemma~\ref{lem-find-ni}. The only difference is that we need to show that for each time slot $T$, the designated channel feedback $m_{i,t} \in  \{\lambda_N, \lambda_S\}$ occurs with probability $1 - 1/\poly(n_i)$, given that the event $\mathcal{P}_{i}^{(t)}$ occurs; this can be achieved via a proof similar to that of Lemma~\ref{lem-aux-1}.

\begin{lemma}\label{lem-find-ni-2}
For an  execution of  $\mathcal{A}$ on $n_i$ devices, $\mathcal{P}_{i}$ holds with probability at least $1 -  n_i^{-7}$.
\end{lemma}

A difference between \strong\ and \weak\ is that in the \strong\ model it is possible to have terminal states in $\Tideal$.
However, there is a simple sufficient condition to guarantee that a state $u$ in  $\Tideal$ is not a terminal state.

\begin{lemma}\label{lem-term-state-2}
Let $u$ be any state in both $\mathcal{T}_i$ and $\mathcal{T}_j$ for some $i \neq j$.
Then $u$ is not a terminal state.
\end{lemma}
\begin{proof}
Suppose that $u$ is a terminal state. The decision of  $u$ is a correct estimate of at most one of $\{n_i, n_j\}$.
Without loss of generality, assume that the decision of $u$ is an incorrect estimate of $n_j$.
When $\mathcal{A}$ is executed on $n_{j}$ devices, any device entering $u$ gives a wrong estimate of $n_{j}$. By Lemma~\ref{lem-find-ni-2}, with probability $1 -  n_{j}^{-7} > 1/{n_{j}}$, there is a device entering $u$, and hence the algorithm fails with probability higher than $1/{n_{j}}$, a contradiction.
\end{proof}

Let $k \geq 3$ be an integer. Consider the set $\{n_1, \ldots, n_k\}$. Our goal is to find an index $\hat{i}$ such that, during an execution of $\mathcal{A}$ on $n_{\hat{i}}$ devices, with probability $1-1/\poly(n_{\hat{i}})$, there exists a device that uses $\Omega(\log k)$ energy. This is achieved by constructing a {\em high energy path} $P=(u_1, u_2, \ldots, u_{\hat{t}})$, along with a sequence of sets of {\em active indices} $K_1 \supseteq K_2 \supseteq \ldots \supseteq K_{\hat{t}}$ in such a way that $i \in K_t$ implies $u_t \in \mathcal{T}_i$. The path $P$ is a directed path in the tree $\Tideal$, and $u_t$ belongs to layer $t$, for each $t$. The number $\hat{t}$ will de chosen later. We will later see that any device entering the state $u_t$ is unable to distinguish between $\{n_i\}_{i \in K_t}$.
The path $P$ is selected to contain at least $\Omega(\log k)$ transitions corresponding to  \listen{} or \send. Thus, choosing $\hat{i}$ as any index in $K_{\hat{t}}$ implies $u_{\hat{t}} \in \mathcal{T}_{\hat{i}}$, and so $\hat{t} \leq d_{\hat{i}}$.  Then, Lemma~\ref{lem-find-ni-2} and the definition of $\mathcal{P}_{i}$ imply that in an execution of $\mathcal{A}$ on $n_{\hat{i}}$ devices, with probability $1-n_{\hat{i}}^{-7}$, at least $n_{\hat{i}} \cdot p_{u_{\hat{t}}} - (\hat{t}-1) \cdot n_{\hat{i}}^{0.6} = \Omega(n_{\hat{i}}^{0.9})> 1$ device enters the state $u_{\hat{t}}$ along the path $P$, and any such device uses $\Omega(\log k)$ energy.

One may attempt to construct the path $P$ by a greedy algorithm which iteratively extends the path by choosing the child state $v$ with the highest probability estimate $p_v$. The ``regular update'' in our construction of $P$ is based on this strategy. However, this strategy alone is insufficient to warrant any energy expenditure in $P$.

We briefly discuss how we force energy expenditure.
Recall that (i) $i \in K_t$ implies $u_t \in \mathcal{T}_i$, and (ii) any device entering $u_t$ is unable to distinguish between the network sizes in $\{n_i\}_{i \in K_t}$. Suppose $i \in K_{d_i}$, and let $s$ be any device in the state $u_{d_i}$. The probability that $s$ remains \idle{} throughout the time interval $[d_i, d_{i+1}-1]$ must be {\em small}. The reason is that $s$ needs to learn whether the underlying network size is $n_i$ by the time $T(n_i) < d_{i+1}$. Suppose that $(u_1, \ldots, u_{t})$ have been defined, and we have $t = d_i$ and $i \in K_{d_i}$. Then it is possible to extend $(u_1, \ldots, u_{t})$ in such a way that guarantees one energy expenditure in the time interval $[d_i, d_{i+1}-1]$. This corresponds to the  ``special update'' in our construction of $P$.

\paragraph{Construction of the High Energy Path.} The path $P = (u_1, \ldots, u_{\hat{t}})$ and the sequence $K_1 \supseteq K_2 \supseteq \cdots \supseteq K_{\hat{t}}$ are defined as follows.
We initialize $\tilde{P}=(u_1)$ with $u_1$ being the initial state, and let $K_1=\{1,2,\ldots,k\}$.

\paragraph{Stopping Criterion.}
The following update rules are applied repeatedly to extend the current path $\tilde{P}=(u_1, \ldots, u_t)$ to a longer path $(u_1, \ldots, u_{t'})$ (for some $t' > t$) until the {\em stopping criterion} $|K_{t}| < 4$ is reached. Then, we set $P = \tilde{P}$.  We will later see in the calculation of the shrinking rate of $|K_t|$ in the proof of Lemma~\ref{thm-rand-cd} that the stopping criterion implies $|K_{t''}| \geq 2$ for all $1 \leq t'' \leq \hat{t}$ in the final path ${P}=(u_1, \ldots, u_{\hat{t}})$. By Lemma~\ref{lem-term-state-2}, this implies that all states in $P$ are not terminal states.


\begin{description}
\item [Regular Update.] We apply this rule if $t\neq d_i$ for all $i \in K_{t}$. Let $x^\star \in \{\send,\listen,\idle\}$ be chosen to maximize $\pp{u_t}{x^\star}$. If $x^\star = \idle$, append the  child of $u_t$ that corresponds to performing $x^\star$ at time $t$ to the end of $\tilde{P}$, and set $K_{t+1}=K_t$.
    In what follows, suppose $x^\star \in \{\send,\listen\}$. If $x^\star = \send$, let $m^\star = \lambda_N$. If $x^\star = \listen$, let $m^\star \in \{\lambda_S,\lambda_N\}$ be chosen to maximize the number of indices $j\in K_{t}$ with $m_{j,t}= m^\star$. Append the child of  $u_t$ that corresponds to performing action $x^\star$ and receiving feedback $m^\star$ at time $t$ to the end of $\tilde{P}$, and set $K_{t+1} = \{ j\in K_{t} \ | \ m_{j,t}= m^\star\}$.
\item [Special Update.] We apply this rule if $t = d_i$ for some $i \in K_{t}$. Let $t' \in \{d_i + 1, \ldots, d_{i+1}\}$ and $x^\star \in \{\send,\listen\}$ be chosen to maximize the probability for a device currently in $u_t$ to be \idle{} throughout the time interval $[t,  t'-2]$ and to perform $x^\star$ at time $t'-1$. If $x = \send$, let $m^\star = \lambda_N$. Otherwise, let $m^\star \in \{\lambda_S,\lambda_N\}$ be chosen to maximize the number of indices $j\in K_{t} \setminus \{i\}$ with $m_{j,t'-1}= m^\star$. We let $u_{t'}$ be the unique descendant of $u_t$ resulting from applying $t'-t$ \idle{} actions throughout the time interval $[t,  t'-2]$ and then performing action $x^\star$ and receiving feedback $m^\star$ at time  $t' - 1$.
The path $\tilde{P}=(u_1, \ldots, u_t)$ is extended to $\tilde{P}=(u_1, \ldots, u_{t'})$. For each $t'' \in \{t+1, \ldots, t'-1\}$, we set $K_{t''} = K_t \setminus \{i\}$. For the new endpoint $u_{t'}$, we set $K_{t'} = \{ j\in K_{t}  \setminus \{i\} \ | \ m_{j,t'-1}= m^\star\}$.
\end{description}
See Figure~\ref{fig-path} for an illustration of the update rules. The reason that $i$ must be removed from the set of the active indices in a special update is that $\mathcal{T}_i$ only contains states up to layer $d_i$. In what follows, we prove properties of the high energy path $P = (u_1, \ldots, u_{\hat{t}})$ resulting from the above procedure. For each $t \in \{1, \ldots, \hat{t}\}$, we define the invariant $\mathcal{I}_t$ as $u_t \in \mathcal{T}_i$ for each $i \in K_t$. By Lemma~\ref{lem-term-state-2}, if $\mathcal{I}_t$ holds and $|K_t| \geq 2$, then $u_t$ is not a terminal state.

\begin{figure}[h]
\begin{center}
\includegraphics[width=0.75\textwidth]{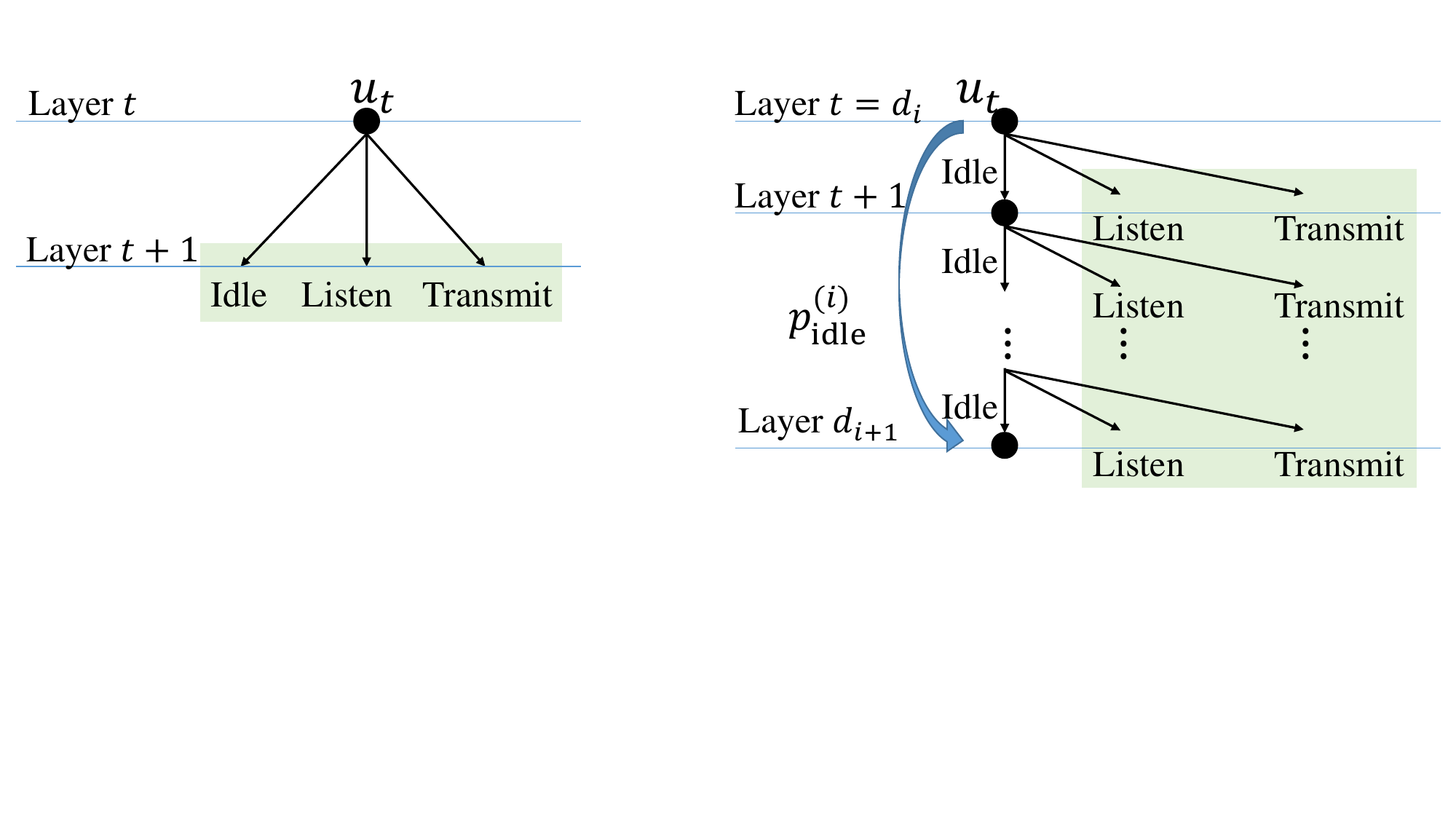}
\end{center}
    \caption{Left: regular update. Right: special update. The shaded region indicates the set of candidate endpoints to extend the current path $\tilde{P}=(u_1, \ldots, u_t)$.
    } \label{fig-path}
\end{figure}

\begin{lemma}\label{lem-idle-prob-2}
Consider a special update for $\tilde{P}=(u_1, \ldots, u_t)$, where $t = d_i$. Let $s$ be a device in $u_t$. Let $p_{\operatorname{idle}}^{(i)}$ be the probability that $s$ remains \idle{} throughout the time interval $[d_i, \ldots, d_{i+1}-1]$. Suppose that $\mathcal{I}_t$ holds. Then $p_{\operatorname{idle}}^{(i)} < 1/2$.
\end{lemma}

\begin{proof}
Since $\mathcal{I}_t$ holds and $i \in K_t$, the state $u_{d_i}$ belongs to $\mathcal{T}_i$.
Lemma~\ref{lem-find-ni-2} implies that with probability $1 - n_i^{-7}$ there is a device $s$ in the state $u_t$ when we execute $\mathcal{A}$ on $n_i$ devices. With probability $p_{\operatorname{idle}}^{(i)}$,  such a device $s$ violates the time constraint $T(n_j)$, since $T(n_j) < d_{j+1}$.
Thus, we must have  $1 / n_i \geq (1 - n_i^{-7})p_{\operatorname{idle}}^{(i)}$, which implies $p_{\operatorname{idle}}^{(i)} < 1/2$.
\end{proof}

\begin{lemma}\label{lem-prob-esti-2}
Let $t \in \{1, \ldots, \hat{t}\}$, and let $n_i \in K_t$.
If  $\mathcal{I}_{\bar{t}}$ holds for all $\bar{t} \in \{1, \ldots, t-1\}$, then $p_{u_t} > n_i^{-1/10}$.
\end{lemma}
\begin{proof}
We first make the following two observations: (i) in a regular update for $\tilde{P}=(u_1, \ldots, u_t)$, we have $p_{u_{t+1}} \geq \frac{p_{u_{t}}}{3}$;
(ii) in a special update for $\tilde{P}=(u_1, \ldots, u_t)$ with $t = d_j$, we have $p_{u_{t'}} \geq \frac{p_{u_t}}{2 \cdot (d_{j+1}-d_j)}  \left(1 - p_{\operatorname{idle}}^{(j)} \right) > \frac{p_{u_{t}}}{4 d_{j+1}}$. Recall that $u_{t'}$ is the new endpoint of $\tilde{P}$ after the special update. Refer to Lemma~\ref{lem-idle-prob-2} for the definition of $p_{\operatorname{idle}}^{(i)} < 1/2$.

Now, fix any $t \in \{1, \ldots, \hat{t}\}$ and $n_i \in K_t$. Notice that $d_i \geq t$ and $i = O(\log^\ast d_i)$.
We write $N_r < t$ and $N_s < i$ to denote the number of regular updates and special updates during the construction of the first $t-1$ states of $P$.
In view of the above two observations, we have:
$$
p_{u_{t}} \geq (1/3)^{N_r} \cdot (1/ 4 d_i)^{N_s}
>  (1/3)^{t} \cdot (1/ 4 d_i)^{i}
 > \left(2^{2^{d_i}}\right)^{-1/10} > n_i^{-1/10},
$$
as desired.
\end{proof}

\begin{lemma}\label{lem-send-r}
Consider a regular update for $\tilde{P}=(u_1, \ldots, u_t)$, and consider any $j \in K_t$.
Suppose that $u_t \in \mathcal{T}_j$, and  $\mathcal{I}_{\bar{t}}$ holds for all $\bar{t} \in \{1, \ldots, t\}$.
if $x^\star = \send$, then $m^\star = \lambda_N = m_{j,t}$.
\end{lemma}
\begin{proof}
By Lemma~\ref{lem-prob-esti-2}, $p_{u_{t+1}} =  p_{u_{t}} \cdot \pp{u_t}{\send} > n_j^{-1/10}$.
Since $u_{t}$, the parent of $u_{t+1}$, is already in $\mathcal{T}_j$, according to the definition of $\mathcal{T}_j$, we must add $u_{t+1}$ to $\mathcal{T}_j$, and set $\lambda_N = m_{j,t}$.
\end{proof}

\begin{lemma}\label{lem-send-s}
Consider a special update that extends $\tilde{P}=(u_1, \ldots, u_t)$ to $(u_1, \ldots, u_{t'})$,
and consider any $j \in K_{t'}$.
Suppose that  $u_{t'-1} \in \mathcal{T}_j$, and $\mathcal{I}_{\bar{t}}$ holds for all $\bar{t} \in \{1, \ldots, t'-1\}$.
If $x^\star = \send$, then $m^\star = \lambda_N = m_{j,t'-1}$.
\end{lemma}
\begin{proof}
The proof is the same as that of Lemma~\ref{lem-send-r}.
By Lemma~\ref{lem-prob-esti-2}, $p_{u_{t'}} =  p_{u_{t'-1}} \cdot \pp{u_t}{\send} > n_j^{-1/10}$.
Since $u_{t'-1}$, the parent of $u_{t'}$, is already in $\mathcal{T}_j$, according to the definition of $\mathcal{T}_j$, we must add $u_{t'}$ to $\mathcal{T}_j$, and set $\lambda_N = m_{j,t'-1}$.
\end{proof}

\begin{lemma}\label{lem-prob-LB}
For each $t \in \{1, \ldots, \hat{t}\}$, $\mathcal{I}_t$ holds.
\end{lemma}
\begin{proof}
For the base case, $\mathcal{I}_1$ holds trivially.
Assume that $\mathcal{I}_{\bar{t}}$ holds for all $\bar{t} \in \{1, \ldots, t-1 \}$, we prove that $\mathcal{I}_t$ holds.
For any $j \in K_t$, we show that $u_t \in \mathcal{T}_j$.

Suppose that $u_t$ is resulting from applying action $x$ and hearing the channel feedback $m$.
By Lemma~\ref{lem-prob-esti-2}, $p_{u_{t}} =  p_{u_{t-1}} \cdot \pp{u_t}{x} > n_j^{-1/10}$.
Since $K_{t} \subseteq K_{t-1}$, by induction hypothesis, $u_{t-1} \in \mathcal{T}_j$.
In what follows, we do a case analysis for all choices of $x \in \{\send, \listen, \idle\}$.

If $x = \idle$, then $u_{t}$ must be in $\mathcal{T}_j$, regardless of the choice of $m_{j, t-1}$, according to the definition of $\mathcal{T}_j$.
If $x = \listen$, then according to the construction of $P$, we have $m = m_{j, t-1}$, and so $u_{t}$ is in $\mathcal{T}_j$.
If $x = \send$, we have $m = \lambda_N$ by the construction of $P$, and $m = \lambda_N = m_{j, t-1}$ due to Lemma~\ref{lem-send-r} and Lemma~\ref{lem-send-s}, and so $u_{t}$ is in $\mathcal{T}_j$.
\end{proof}

We are now in a position to prove the main result of this section.

\begin{lemma}\label{thm-rand-cd}
For any positive integer $k$, there is a network size $n$ satisfying $d_1 \leq n \leq d_{k+1}$ such that in an execution of $\mathcal{A}$  on $n$ devices, with probability at least $1 - n^{-7}$, there is a device that performs $\Omega(\log k)$ \listen{} steps.
\end{lemma}
\begin{proof}
First, we bound the shrink rate of the size of active indices $K_t$.
Consider a regular update for $\tilde{P}=(u_1, \ldots, u_t)$. If $x^\star = \idle$, then $K_{t+1} = K_t$.
If  $x^\star = \send$, then we also have $K_{t+1} = \{ j \in K_{t} \ | \ m_{j,t}= m^\star\} = K_t$ in view of Lemma~\ref{lem-send-r}.
If  $x^\star = \listen$, then our choice of $m^\star$ in the regular update implies $|K_{t+1}| \geq |K_t| / 2$.
Next, consider a special update that extends $\tilde{P}=(u_1, \ldots, u_t)$ to $(u_1, \ldots, u_{t'})$.
Similarly, if  $x^\star \in \{\idle, \send\}$, then $K_{t'} = K_t \setminus \{i\}$, where $i$ is the index such that $t = d_i$; see Lemma~\ref{lem-send-s}.
For the case of $x^\star = \listen$,  our choice of $m^\star$ in the special update implies $|K_{t'}| \geq  \left(|K_t| - 1 \right) / 2$.
Therefore, any device whose execution history following in the path $P=(u_1, \ldots, u_{\hat{t}})$ performs $\Omega(\log |K_1| - \log |K_{\hat{t}}|)$ \listen{} steps.

The stopping criterion,
together with our calculation of the shrinking rate of $|K_t|$, implies that $|K_{\hat{t}}| \geq 2$.
We let $\hat{i}$ be any element in $K_{\hat{t}}$, and set $n = n_{\hat{i}}$.
By Lemma~\ref{lem-prob-LB},  $u_{\hat{t}} \in \mathcal{T}_{\hat{i}}$.
Then, Lemma~\ref{lem-find-ni-2} implies that in an execution of $\mathcal{A}$ on $n_{\hat{i}}$ devices, with probability $1-n_{\hat{i}}^{-7}$, at least $n_{\hat{i}} \cdot p_{u_{\hat{t}}} - (\hat{t}-1) \cdot n_{\hat{i}}^{0.6} = \Omega(n_{\hat{i}}^{0.9})> 1$ devices enter the state $u_{\hat{t}}$ along the path $P$, and any such device  performs $\Omega(\log |K_1| - \log |K_{\hat{t}}|) = \Omega(\log k)$ \listen{} steps.
\end{proof}

Similarly, so long as $T(n) \leq \exp^{(\ell)}(n)$, for some constant $\ell$, where $\exp^{(i)}$ is iterated $i$-fold application of $\exp$, it is possible to set the checkpoints such that $k = \Theta( \log^\ast (d_{k+1}))$, and so Lemma~\ref{thm-rand-cd} implies that the energy cost $\mathcal{A}$ is  $\Omega(\log \log^\ast n)$. Therefore, we conclude Theorem~\ref{thm:randLB-2}.

\subsection{Other Problems}
\label{sec:otherlb}
In this section we discuss lower bounds of other problems.

\paragraph{Successful Communication.}
We demonstrate how our lower bounds proofs  can be adapted to the class $\mathcal{C}$ of all problems that require each
device to perform at least one successful communication before it terminates.
In particular, this includes  \LeaderElection{} and the contention resolution problem studied in~\cite{BenderKPY16}.
Notice that \ApproximateCounting, in general, does not require each device to perform a successful communication before it terminates.

Consider the \weak{} model. Let $\mathcal{A}$ be a polynomial time algorithm, and a device in an execution of $\mathcal{A}$ must perform at least one successful communication before it terminates. Let the runtime of $\mathcal{A}$ be $T(n)$. Let $n = n_i$ for some $i$.
Consider a device $s$ in an execution of $\mathcal{A}$ on $n_i$ devices.
Let $t_{\operatorname{suc}}$ be time of the first successful communication of $s$.
Then $t_{\operatorname{suc}} \leq  T(n_i) < d_{i+1}$ with probability $1 - 1/n_i$.
By Lemma~\ref{lem-find-ni}, with probability $1 - n_i^{-7}$, by time $d_i$ all devices are confined in $\mathcal{T}_i \subseteq \Tideal$ and no successful communication occurs throughout the time interval $[1, d_i - 1]$.
Thus, $t_{\operatorname{suc}} \geq d_i$ with probability $1 - n_i^{-7}$.
Since $t_{\operatorname{suc}}$ is within $[d_i, d_{i+1})$ with probability $1 - 1/\poly(n_i)$, this number can be seen as a very loose estimate of the network size $n = n_i$, but this estimate is already good enough for the device $s$ to distinguish $n = n_i$ from other candidate network sizes in $\{n_j\}$. Since we only consider the set of network sizes $\{n_j\}$ in our proof for Theorem~\ref{thm:randLB-1}, the proof applies to $\mathcal{A}$. For the same reason, Theorem~\ref{thm:randLB-2} also applies to all problems in the class $\mathcal{C}$ in \strong.

\paragraph{Loneliness Detection.}
We consider the  loneliness detection problem whose goal is to distinguish between $n=1$ and $n > 1$; see~\cite{GhaffariN16,GhaffariLS12}.
We show that this problem is impossible to solve in \nocd.
Intuitively, in \nocd, a transmitter cannot simultaneously listen to the channel, and so a device  never receives any feedback from the channel if $n=1$. However, when $n$ is large enough relative to $t$, with high probability a device also does not hear any message in the first $t$ time slots. It seems hopeless to have an algorithm that detects loneliness.

Let $T(n)$ be any time function.
Let $\mathcal{A}$ be any algorithm in \nocd\  that accomplishes the following. If $n > 1$, with probability at least $1 - 1 / n$,  all devices terminate by time $T(n)$ and output ``$n > 1$''. If $n = 1$, then the only participating device $s$ terminates by time $t$ and outputs ``$n = 1$'' with probability $p$. We show that either $t = \infty$ or $p = 0$.

We simulate $\mathcal{A}$ in \weak\ and apply the analysis in Section~\ref{sec:weakcdlb}. 
Recall that in \nocd\ a transmitter cannot simultaneously listen to the channel, and so
for each terminal state $u \in \mathcal{T} \setminus \Tideal$ such that the path $P$ leading to $u$ does not involve successfully listening to a message,  the output of $u$ is identical to some state $u' \in \Tideal$ (which results from changing each successful transmission to a failed transmission in the execution history).

For each state $u \in \Tideal$, there exists an index $i$ such that $u \in \mathcal{T}_i$.
 By Lemma~\ref{lem-find-ni}, in an execution of $\mathcal{A}$ on $n_i$ devices, with probability $1 - n_i^{-7}$, there is at least one device entering the state $u$. Thus, no state in $\Tideal$ is a terminal state with output ``$n = 1$''. However, in \nocd\ with $n=1$ there is no means for a listener to receive a message, and so we must have either $t = \infty$ or $p = 0$.


\newcommand{\bottomup}{\textsf{SimpleCensus}}
\newcommand{\legal}{centralized}
\newcommand{\merge}{\textsf{Merge}}
\newcommand{\phase}{\textsf{DetLE}}

\section{Deterministic Upper Bound}\label{sct:detalg}

In this section we present an optimal deterministic algorithm for \Census{} in \weak{} that simultaneously matches the $\Omega(\log \log N)$ energy lower bound of Theorem~\ref{thm:det_lb} and the $\Omega(N)$ time lower bound of~\cite[Theorem~1.6]{JurdzinskiKZ02}.  Notice that any  \Census{} algorithm also solves \LeaderElection.

\begin{theorem}\label{thm:detUB}
There exists a deterministic \weak{} algorithm
that solves  \Census{}
in $O(N)$ time with energy $O(\log \log N)$.
\end{theorem}


Our algorithm is inspired by an energy-sharing technique introduced~\cite{JurdzinskiKZ02}, which is based on the concept of groups.
We call an ID {\em active} if there is a device of such an ID; we also write $s$ to denote the device of ID $s$.

A {\em group} $G$ is a set of active IDs meeting the following criteria.
Each device belongs to at most one group.
Let $G = (s_1, \ldots, s_k)$ be the members of $G$, listed in increasing order by ID.
The {\em rank} of a device $s_i \in G$ is defined as $i$.
We assume each group $G$ has a unique \emph{group ID}.
Similarly, we say that a group ID $x$ is {\em active} if there is a group $G$ whose ID is $x$.
Each group $G$ has a device $s \in G$ that serves as the {\em representative} of $G$.
We allow the representative of a group to be changed over time.
Each device $s \in G$ knows (i) the group ID of $G$, (ii) the current representative of $G$, and (iii) the list of all IDs in $G$.




\subsection{A Simple Census Algorithm}
In this section we show how to use groups to distribute energy costs to devices.
Consider the following setting. All devices are partitioned into groups whose IDs are within the range $\{1, \ldots, \hat{N}\}$, and each group has size at least $g$. We present a simple \nocd\ deterministic algorithm $\bottomup(\hat{N},g)$ that elects a leader group $G^\star$ such that the representative of $G^\star$ knows the list of IDs of all devices. The algorithm $\bottomup(\hat{N},g)$ ends when the representative of $G^\star$ announces the list of IDs of all devices; the last time slot of $\bottomup(\hat{N},g)$ is called the ``announcement time slot.''
The algorithm $\bottomup(\hat{N},g)$ is executed recursively, as follows.

\paragraph{Base Case.} If $\hat{N}=1$, then there is only one group $G$ which contains all devices. By definition, each device in $G$ already knows the list of IDs of all devices in $G$.
We set $G^\star = G$, and let the representative of $G$ announce the list of IDs of all devices in $G$ at the announcement time slot.

\paragraph{Inductive Step.} Assume $\hat{N}>1$. The algorithm invokes two recursive calls.
For each group $G$, if the rank of the current representative of $G$ is $i$, then the representative of $G$ in a recursive call will be the device of rank $i+1$ (or 1 if $i = |G|$).
The two recursive calls are made on the two halves
of the ID space, $S_1 = \{1, \ldots, \lceil \hat{N}/2 \rceil\}$ and $S_2 = \{\lceil \hat{N}/2 \rceil + 1, \ldots, \hat{N}\}$.
The representative $s$ of each group $G$ listens to the announcement time slots of the two recursive calls; after that, $s$ learns the list of IDs of all devices.
The leader group $G^\star$ is selected as the one that contains the device of the smallest ID.
Then we let the representative of $G^\star$ announce the list of IDs of all devices at the announcement time slot.

\medskip

It is straightforward to see that the algorithm $\bottomup(\hat{N},g)$ takes $O(\hat{N})$ time and uses $O(\lceil \frac{\log \hat{N}}{g} \rceil)$ energy. We summarize the result as a lemma.

\begin{lemma}
\label{thm_simplele}
Suppose that all devices are partitioned into groups whose IDs are within the range $\{1, \ldots, \hat{N}\}$, and each group has size at least $g$.
The algorithm $\bottomup(\hat{N},g)$ elects a leader group $G^\star$ such that the representative of $G^\star$ knows the list of IDs of all devices. The algorithm  takes $O(\hat{N})$ time and uses $O(\lceil \frac{\log \hat{N}}{g} \rceil)$ energy.
\end{lemma}

\subsection{An Optimal Census Algorithm}\label{sct:detle}
In this section we prove Theorem~\ref{thm:detUB}.
Without loss of generality, we assume that $\log N$ is an integer. Our algorithm consists of $O(\log \log N)$ {phases}. All devices participate initially and may drop out in the middle of the algorithm. We maintain the following invariant $\mathcal{I}_i$ at the beginning of the $i$th phase.

\paragraph{Invariant $\mathcal{I}_i$.} (i) All participating devices are partitioned into groups of size exactly $2^{i-1}$ in the group ID space $\{1, \ldots, N\}$. (ii) For each device that drops out during the first $i-1$ phases, its ID is remembered by the representative of at least one group.


\paragraph{Termination.}  There are two different outcomes of our algorithm. For each index $i$, the algorithm terminates at the beginning of the $i$th phase if either (i) there is only one group remaining, or (ii) $i = \log \log N + 1$.

Suppose that at the beginning of the $i$th phase there is only one group $G$ remaining; then the algorithm is terminated with the representative of $G$ knowing the list of IDs of all devices.
Suppose that more than one group remains at the beginning of the $(\log \log N + 1)$th phase; as the groups that survive until this moment have size $\log N$, we can apply $\bottomup(N, \log N)$ to solve \Census\ in $O(N)$ time with $O(1)$ energy cost.

\paragraph{Overview.} At the beginning of the first phase, each device $s$ forms a singleton group $G = \{s\}$, and so the invariant  $\mathcal{I}_1$ is trivially met. Throughout the algorithm, the group ID of a group $G$ is always defined as the minimum ID of the devices in $G$.

During the $i$th phase, each group attempts to find another group to merge into a group with size $2^{i}$. Each group $G$ that is not merged drops out, and the list of IDs in $G$ is remembered by the representative of at least one other group $G'$ that does not drop out.
In what follows, we describe the algorithm for the $i$th phase.  At the beginning of  the $i$th phase, it is guaranteed that the invariant  $\mathcal{I}_i$ is met. We also assume that the number of groups is at least 2, since otherwise the terminating condition is met.

\paragraph{Step~1---Merging Groups.} The first step of the $i$th phase consists of the procedure $\phase(N)$, which costs $O(N)$ time and $O(\log \log N)$ energy.  For each device $s$, we write $I(s)$ to denote the set of all IDs that the device $s$ has heard during the first $i-1$ phases (including the ID of $s$). Notice that $I(s) \supseteq G$ if $s$ belongs to the group $G$.
The procedure $\phase(\hat{N})$ is defined recursively as follows.

\paragraph{Base Case.} Suppose that the group ID space $S$ has size $\hat{N} = 2$, and there are exactly two groups $G_1$ and $G_2$. Using two time slots, the representatives $s_1$ and $s_2$ of the two groups exchange the information $I(s_1)$ and $I(s_2)$, and then the two groups are merged.

\paragraph{Inductive Step.} Suppose that the group ID space $S$ has size $\hat{N} > 2$, and there are at least two groups. Uniformly divide the group ID space $S = \{1, \ldots, \hat{N}\}$ into $N' = \lceil\sqrt{\hat{N}}\rceil$ intervals $S_1, \ldots, S_{N'}$, and each of them has size at most $N' = \lceil\sqrt{\hat{N}}\rceil$.

For each $j\in \{1, \ldots, N' \}$, let $z_j$ be the number of groups whose ID is within $S_j$.
In the \weak\ model, testing whether $z_j = 1$ can be done by letting the representatives of all groups
whose ID are in $S_j$ speak simultaneously.
If $z_j \neq 1$, invoke a recursive call to $\phase(N')$ on the group ID space $S_j$;  the recursive call is vacuous if $z = 0$.

Let $\mathcal{G}$ be the set of all groups that do not participate in the above recursive calls. That is, $G \in \mathcal{G}$ if the ID of $G$ belongs to an interval $S_j$ with $z_j = 1$.
In the \weak\ model, we can check whether $|\mathcal{G}| = 1$ in one time slot by letting the
representatives of groups in $\mathcal{G}$ speak simultaneously.
For the case of $|\mathcal{G}|  \neq 1$, we invoke a recursive call to $\phase(N')$ on $\mathcal{G}$,
where the ID space is $S' = \{1, \ldots, N'\}$ and the group ID of the group from $S_j$ is $j$.
For the case of $|\mathcal{G}| = 1$, we allocate one time slot to let the representative $s$ of the unique group
$G \in \mathcal{G}$ announce $I(s)$ to the representatives of all other groups, and then the group
$G$ decides to drop out from the algorithm.

\paragraph{Analysis.} By the end of $\phase(\hat{N})$, for each group $G$ whose representative is $s$, we have (i) $G$ is merged with some other group $G'$, or (ii) $G$ drops out, and $I(s)$ is remembered by the representative $s'$ of some other group $G'$.

Let $E(\hat{N})$ and $T(\hat{N})$ denote the energy complexity and the time complexity of $\phase(\hat{N})$ on group ID space of size $\hat{N}$.
We have $T(2)=E(2)=O(1)$ and
\begin{align*}
E(\hat{N}) &= E(\lceil\sqrt{\hat{N}}\rceil) + O(1)\\
T(\hat{N}) &= (\lceil\sqrt{\hat{N}}\rceil + 1) \cdot T(\lceil\sqrt{\hat{N}}\rceil) + O(\lceil\sqrt{\hat{N}}\rceil).
\end{align*}
It is straightforward to show that $E(\hat{N}) = O(\log\log\hat{N})$ and $T(\hat{N}) = O(\hat{N})$.

\paragraph{Step~2---Disseminating Information and Electing New Representatives.} Notice that only the representatives of the groups participate in Step~1. Let $G$ be a group whose representative is $s$. If $G$ is merged with some other group $G'$ whose representative is $s'$, then we need all members in $G$ to know $I(s')$ and the list of members in $G'$. If $G$ decides to drop out from the algorithm, then we need all members in $G$ to know about this information. We allocate $N$ time slots for the representatives to communicate with other group members. The energy cost for the information dissemination is $O(1)$ per device.

To save energy, we need each device to serve as a representative for not too many phases. For the first phase, each device $s$ inevitably serves as the representative of its group $G = \{s\}$. For $i > 1$, for each group $G$ participating in the $i$th phase, we let the device of rank $2^{i-1}$ be the representative of $G$ during the $i$th phase. It is straightforward to verify that each device serves as a representative for at most two phases throughout the entire algorithm.




\paragraph{Time and Energy Complexity.}
We analyze the runtime and the energy cost of the entire algorithm.
The energy cost for a device $s$ in one phase is $O(\log \log N)$ if $s$ serves as a representative, and is $O(1)$ otherwise.
There are $O(\log \log N)$ phases, and each device serves as a representative for no more than two phases throughout the algorithm.  Therefore, the  energy cost of the algorithm is $O(\log \log N)$ per device.
The runtime of the algorithm is $O(N \log  \log N)$, since each phase takes $O(N)$ time.
The runtime can be further reduced to $O(N)$ by doing the following preprocessing step.
Uniformly divide the ID space into $\frac{N}{\log \log N}$ intervals, and call
$\bottomup(\log \log N, 1)$ on each interval.
This takes   $O(N)$ time and $O(\log \log \log N)=o(\log \log N)$ energy.
After the preprocessing, each interval has a leader that knows the list of IDs of all devices in the interval.
We only let the leaders of the $\frac{N}{\log  \log N}$ intervals participate in our algorithm. This reduces the size of the ID space from $N$ to $N' = \frac{N}{\log \log N}$, and so the runtime is improved to $O(N' \log  \log N') = O(N)$.

\newcommand{\densealgo}{\textsf{DenseAlgo}}
\newcommand{\denseinit}{\textsf{DenseInit}}
\newcommand{\groupid}{\textsf{GID}}
\newcommand{\roundid}{\textsf{RID}}
\newcommand{\startid}{\textsf{SID}}
\newcommand{\idspacesize}{\textsf{Size}}
\newcommand{\termin}{\textsc{Terminated}}

\section{Deterministic Upper Bound for Dense Instances}\label{sec:detle_dense}
In this section we present a deterministic algorithm that solves  \Census~with inverse Ackermann energy cost when the input is {\em dense} in the ID space, i.e., the number of devices $n$ is at least $c \cdot N$ for a fixed constant $c > 0$. This improves upon a prior work of Jurdzinski et al.~\cite{JurdzinskiKZ02} which uses $O(\log^\ast N)$ energy.
For any two positive integers $i$ and $j$, we define the two functions $a_i(j)$ and $b_i(j)$ as follows.

\begin{center}
\begin{tabular} { c c c  }
$
a_i(j) =
\begin{cases}
j^9 &\text{if $i = 1$}\\
a_{i-1}^{(j)}(j^8)&\text{if $i > 1$}
\end{cases}
$
&
\hspace{1cm}
&
$
b_i(j) =
\begin{cases}
2^j &\text{if $i = 1$}\\
2^j \prod_{r = 0}^{j-1} b_{i-1} \left (a_{i-1}^{(r)}(j^8)\right) &\text{if $i > 1$}
\end{cases}
$
\end{tabular}
\end{center}
The notation $f^{(r)}$ is iterated $r$-folded application of $f$, which is defined as $f^{(0)}(x) = x$ and $f^{(r)}(x) = f\left(f^{(r-1)}(x)\right)$.
We define the inverse Ackermann function $\alpha(N)$ to be the minimum number $i$ such that $b_i(55) \geq N$.
This is not the standard definition of $\alpha$, but it is identical to any other definition from the literature, up to $\pm O(1)$.
The goal of this section is to prove the following theorem.

\begin{theorem}\label{thm:dense}
Suppose the number of devices $n$ is at least $c \cdot N$ for a fixed constant $c > 0$.
There is  a deterministic \nocd{} algorithm that solves  \Census\ in time $O(N)$ with energy cost $O(\alpha(N))$.
\end{theorem}


Our algorithm is based on the recursive subroutine $\densealgo_i(\hat{N}, j)$, which is capable of merging groups into fewer and larger ones using very little energy. The parameter $i$ is a positive integer indicating the height of the recursion. The parameter $\hat{N}$ is an upper bound on the size of the group ID space; for technical reasons we allow $\hat{N} \geq 1$ to be a fractional number. The parameter $j$ is a lower bound on the group size; we assume $j \geq 55$.
The precise specification of  $\densealgo_i(\hat{N}, j)$ is as follows.


\begin{description}
\item [Input.] Prior to the execution of $\densealgo_i(\hat{N}, j)$, the set of all devices are partitioned into groups. The size of the group ID space is at most $\hat{N}$. Each group has size at least $j$. The total number of groups is at least $\hat{N} / \log j$.

\item [Output.] Some devices drop out during the execution of $\densealgo_i(\hat{N}, j)$. The fraction of the devices that drop out is at most $2/j$ of all devices. After the execution of $\densealgo_i(\hat{N}, j)$, the remaining devices form new groups. The size of the group ID space is at most $N' = \max\{1,  \hat{N} / b_i(j) \}$.
     Each group has size at least $a_i(j)$. The total number of groups is at least $N' / {\log a_i(j)}$.
    We allocate  $N'$  ``announcement time slots'' at the end of $\densealgo_i(\hat{N}, j)$.  At the $k$th announcement time slot, the representative of the group $G$ of ID $k$ announces the list of all members of $G$.

    We have more stringent requirements for the case of $i = 1$:  (i) the fraction of the devices that drop out is at most $1/j$ of all devices, and (ii) the total number of groups is at least $N' / {8 \log j}$.

\item [Complexity.] The procedure $\densealgo_i(\hat{N}, j)$ takes $O(\hat{N})$ time and consumes $O(i)$ energy per device.
\end{description}

In Sections~\ref{densealgo-base} and~\ref{densealgo-inductive}, we present and analyze the subroutine $\densealgo_{i}(\hat{N}, j)$.
We have the following auxiliary lemma.

\begin{lemma}\label{lem:census-reduce}
Suppose the ID space of devices is $S = \{1, \ldots, N\}$, and there is a group $G$ of size $\epsilon N$.
Then there is a deterministic algorithm that solves \Census\ in $O(N)$ time and $O(1 / \epsilon)$ energy.
\end{lemma}

\begin{proof}
The algorithm is as follows.
Partition the ID space $S$ into $k = \epsilon N = |G|$ intervals $S_1, \ldots, S_k$, where each interval has size at most $\ceil{1/\epsilon}$.
Let $s_i$ be the device in $G$ that is of rank $i$, and let $L_i$ be the list of IDs of all devices in $S_i$.
For each $1 \leq i \leq k$, we let $s_i$ learn $L_i$ by having $s_i$ listens for $|S_i| = O(1 / \epsilon)$ time slots, where each device in $S_i$ transmits once.
Next, we allocate $k-1$ time slots to do the following. For $i = 1$ to $k-1$, Let $s_i$ transmit $\bigcup_{j=1}^{i} L(s_j)$ and $s_{i+1}$ listen.  After that, $s_k$ knows the list of IDs of all devices, and we let $s_k$ announce the list while all other devices listen to the channel.
\end{proof}

We are now in a position to prove Theorem~\ref{thm:dense}.
The proof is based on Lemma~\ref{lem:census-reduce} and the procedure  $\densealgo_i(\hat{N}, j)$.
Recall that the number of devices is promised to be at least $cN$.
We choose $j^\ast = \max\{55, 2^{\ceil{1/c}}\}$ and let $i^\ast$ be the minimum number $i$ such that $N / b_i(j^\ast) \leq 1$.
In order to artificially satisfy the input invariant, we imagine that each device simulates a group of $\log j^\ast = O(1)$ devices.
Notice that the group ID space is $\{1, \ldots, N\}$, and the total number of groups is $n \geq c N \geq N / \log j^\ast$.
Thus, the requirement for executing $\densealgo_{i^\ast}({N}, j^\ast)$ is met.
We execute $\densealgo_{i^\ast}({N}, j^\ast)$, which costs $O(N)$ time and $O(i^\ast) = O(\alpha(N))$ energy.
During the execution, at most $2/j^\ast$ fraction of devices drop out.
All remaining devices form $1 = \max \{1,  {N} / b_{i^\ast}(j^\ast) \}$ group.
After that, we can solve \Census\  by the algorithm of Lemma~\ref{lem:census-reduce} using additional $O(N)$ time and $O(1)$ energy.

\subsection{Base Case}
\label{densealgo-base}
In this section we present the base case $i = 1$ of the subroutine
 $\densealgo_i(\hat{N}, j)$ and show that it meets the required specification.
At the beginning, all devices are organized into groups of size at least $j$, and the group ID space $S$ has size at most $\hat{N}$.
We partition the group ID space $S$ into $k = \lceil \hat{N} / 2^{j+1} \rceil$ intervals $S_1, \ldots, S_k$ such that each interval has size at most $2^{j+1}$.
For each interval $S_l$, we run $\bottomup(|S_l|,j)$ to merge all groups in the interval $S_l$ into a single group $G$, and let $l$ be the ID of $G$.
After that, for each group $G$ of size less than $j^9$, all devices in $G$ drop out.
The execution of $\bottomup(|S_l|,j)$ takes $O(|S_l|)$ time and $O(1)$ energy. Thus, algorithm $\densealgo_1(\hat{N}, j)$ costs $O(\hat{N})$ time and $O(1)$ energy. It is clear that each group in the output has size at least $j^9 = a_1(j)$. We show that the output meets the remaining requirements.

\begin{description}
\item[Size of Group ID Space.] The size of the output group ID space is $k = \lceil \hat{N} / 2^{j+1} \rceil$. For the case of $\hat{N} / 2^{j+1} \leq 1$, we have $k = 1$. For the case of $\hat{N} / 2^{j+1} > 1$, we have $k = \lceil \hat{N} / 2^{j+1} \rceil < \hat{N} / 2^{j} = \hat{N} / b_1(j)$. Thus, the size of the group ID space $k$ is always upper bounded by $N' = \max\{1,  \hat{N} / b_1(j) \}$.
\item[Proportion of Terminated Devices.] The number of devices that are terminated is at most $z_{\text{term}} =  (j^9 - 1)k < j^9 \hat{N} / 2^{j}$. The total number of devices is at least $z_{\text{init}} = j \hat{N} / \log j$. Thus, the proportion of the terminated devices is at most $f = z_{\text{term}} / z_{\text{init}} = \frac{j^8  \log j}{2^{j}}$. As long as $j \geq 55$, we have $f < 1/j$.
\item[Number of Groups.] The total number of devices is at least $z_{\text{init}}  = j \hat{N} / \log j$. The size of each output group is at most $j 2^{j+1}$. Since the proportion of the terminated devices is at most $1/j \leq 1/55 < 1/2$, the number of output groups is at least
    $$z_{\text{out}} = \max\{1, \lfloor (z_{\text{init}}/2) / \left(j 2^{j+1}\right) \rfloor\} =
     \max\{1, \lfloor z_{\text{init}} / \left(j 2^{j+2}\right) \rfloor\}.$$

    We show that the inequality $z_{\text{out}} \geq   \max\{1, \hat{N} / (2^{j} 8 \log j)\} \geq N' / {8 \log j}$ holds. For the case of $z_{\text{out}} = 1$, the inequality is already met. If $z_{\text{out}} > 1$, we have
    $z_{\text{out}} = \lfloor z_{\text{init}} / \left(j 2^{j+2}\right) \rfloor
    \geq  z_{\text{init}} / \left(j 2^{j+3}\right) = \hat{N} / (2^{j} 8 \log j)$, as desired.
\end{description}

\subsection{Inductive Step}\label{densealgo-inductive}

In this section we consider the case of $i > 1$.
The algorithm $\densealgo_i(\hat{N}, j)$ begins with
an initialization step, which increases the group size from $j$ to $j^9$ by executing $\densealgo_1(\hat{N}, j)$.
After that, it recursively invokes $\densealgo_{i-1}(X_r, Y_r)$, for $r$ from $1$ to $j^\star$, where $j^\star$
and the sequences $(X_r)_{r \in [j^\star]}$ and $(Y_r)_{r \in [j^\star]}$ will be determined.
Each device participates in the initialization step and exactly one recursive call to $\densealgo_{i-1}$, so
the energy cost per device is $O(1) + O(i-1) = O(i)$.

After the initialization step, each group $G$ has size $j^9$.  For each group $G$, we extract $j^\star$ subgroups $G_1, G_2, \ldots, G_{j^\star}$ from the members of $G$, each
with size exactly $j^8$ (we will later see that $j^\star \leq j$).
The subgroup $G_r$ is responsible for representing $G$ in the $r$th recursive call $\densealgo_{i-1}(X_r, Y_r)$.
For $1 \le r < j^\star$, as $G_{r}$ and $G_{r + 1}$ have the same size, we set up a bijection $\phi_r: G_r \to G_{r + 1}$.
For each device $s$ in the $r$th subgroup $G_r$,
after $s$ finishes the $r$th recursive call, if $s$ has not dropped out yet,
$\phi_r(s)$ continues to play the role of $s$ in the $(r + 1)$th recursive call.
$\phi_r(s)$ learns all information known to $s$ by listening to an announcement time slot of the $r$th recursive call.

\paragraph{Parameters of Recursive Calls.}
If $\hat{N}/2^j < 2$, then only one group remains after the initialization step $\densealgo_1(\hat{N}, j)$, and so we are already done without doing any more recursive calls. In what follows, we assume $\hat{N}/2^j \geq 2$.
The two sequences $(X_r)_{r \in [j^\star]}$ and  $(Y_r)_{r \in [j^\star]}$ are defined as follows.
We choose $j^\star$ as $\min\{j, \arg \min_r (X_r < 2)\}$.\footnote{We will later see that $X_r$ represents the output group ID space of the $r$th recursive call. If $X_r < 2$ for some $r < j$, then we can terminate after the $r$th recursive call.}

\begin{center}
\begin{tabular} { c c c  }
$
X_r =
\begin{cases}
\hat{N}/2^j &\text{if $r = 1$}\\
X_{r-1} / b_{i-1}(Y_{r-1}) &\text{if $r > 1$}
\end{cases}
$
&
\hspace{1cm}
&
$
Y_r =
\begin{cases}
j^8 &\text{if $r = 1$}\\
a_{i-1}(Y_{r-1}) &\text{if $r > 1$}
\end{cases}
$
\end{tabular}
\end{center}

We verify that the requirement of executing the $r$th recursive call is met, for each $1 \leq r \leq j^\star$.
\begin{description}
\item[Base Case.] For $r = 1$, we show that the requirement of $\densealgo_{i-1}(X_1, Y_1)$ is met after the initialization step  $\densealgo_1(\hat{N}, j)$: (i)  the number of groups is at least $\frac{\hat{N}/2^j}{8 \log j} = X_1 / \log Y_1$; (ii) the size of each group is $j^8 = Y_1$; (iii) the group ID space is at most $\hat{N}/2^j = X_1$.
\item[Inductive Step.] For $1 < r \leq j^\star$,   we show that the requirement of $\densealgo_{i-1}(X_r, Y_r)$ is met after the previous recursive call $\densealgo_{i-1}(X_{r-1}, Y_{r-1})$:   (i)  the number of groups is at least $\frac{X_{r-1} / b_{i-1}(Y_{r-1})}{\log a_{i-1}(Y_{r-1})} = X_r / \log Y_r$; (ii) the size of each group is $a_{i-1}(Y_{r-1}) = Y_r$; (iii) the group ID space is at most $X_{r-1}/ b_{i-1}(Y_{r-1}) = X_r$.
\end{description}
It is also straightforward to see that the output (group size, number of groups, and group ID space size) of the last recursive call $\densealgo_{i-1}(X_{j^\star}, Y_{j^\star})$ already satisfies the requirement of the output of $\densealgo_i(\hat{N}, j)$, since we have $X_{j} / b_{i-1}(Y_{j}) = \hat{N} / b_i(j)$ and $a_{i-1}(Y_{j}) = a_i(j)$. Next, we show that the number of  devices that drops out during the execution of $\densealgo_i(\hat{N}, j)$ is at most $2/j$ of all devices.
Let $f_i(j)$ be the fraction of devices the are terminated during the execution of $\densealgo_i(\hat{N}, j)$.
The analysis in Section~\ref{densealgo-base} implies that $f_1(j) \leq \frac{1}{j}$.
We prove that $f_i(j) \leq \frac{2}{j}$.
\begin{align*}
1 - f_i(j)
&\geq  (1 - f_1(j)) \prod_{r = 1}^{j^\star}  (1 - f_{i-1}(Y_r)) \\
&\geq  (1 - 1/j)  \prod_{r = 1}^{j^\star}   (1 - 2 / Y_r) & \text{(by induction hypothesis)} \\
&\geq (1 - 2/j).
\end{align*}

\paragraph{Energy Complexity.} 
During $\densealgo_{i}(\hat{N}, j)$, each device uses $O(1)$ energy in the initialization step.
Consider a device $s$ participating in the $r$th recursive call.
If $r > 1$, then $s$ uses $O(1)$ energy to learn the information of $\phi_{r - 1}^{-1}(s)$.
The execution of  $\densealgo_{i-1}(X_r, Y_r)$ costs $O(i-1)$ energy.
Thus, each device spends $O(i)$ energy during the execution of $\densealgo_i(\hat{N}, j)$.

\paragraph{Time Complexity.}
Let $T_i(\hat{N})$ be the runtime of $\densealgo_{i}(\hat{N}, j)$ (for any $j$).
The analysis in Section~\ref{densealgo-base} implies that $T_1(\hat{N}) \leq C \hat{N}$ for some constant $C$.
We prove that $T_i(\hat{N}) \leq 10 C \hat{N} = O(\hat{N})$, for all $i$.
\begin{align*}
T_i(\hat{N}) 
&= T_1(\hat{N}) + \sum_{r=1}^{j^\star} T_{i-1}(X_r)\\
&\leq C \hat{N} + \sum_{r=1}^{j^\star} (10 C) X_r  & \text{(by induction hypothesis)} \\
&= C \hat{N} + 10C \sum_{r=1}^{j^\star} X_r\\
&\leq  C \hat{N} + 10C \cdot 0.9  \hat{N}\\
&= 10C \hat{N}.
\end{align*}
To summarize, $\densealgo_{i}(\hat{N}, j)$ costs $O(\hat{N})$ time and $O(i)$ energy, and
the constant hidden in $O(\hat{N})$ is an absolute constant independent of $i$.

\newcommand{\Verify}{\textsf{Verify}}
\newcommand{\ExpSearch}{\textsf{ExpSearch}}
\newcommand{\EstimateSize}{\sf EstimateSize}

\section{Randomized Upper Bounds}\label{app-sect:randUB}

In this section we present randomized algorithms for \ApproximateCounting{} matching the energy complexity lower bound proved in Section~\ref{app-sect:randLB}. In~\cite{BenderKPY16}, a randomized algorithm for \ApproximateCounting{} in \strong{} using $O(\log(\log^* n))$ energy is devised. They showed that any circuit of constant fan-in, with input bits encoded as $\noise = 1$ and $\silence = 0$, can be simulated with $O(1)$ energy cost, and an estimate of the network size can be computed by such a circuit.
The circuit simulation of \cite{BenderKPY16} makes extensive use of collision detection.
In this section we demonstrate a different approach to \ApproximateCounting{}
based on our dense \Census{} algorithm, which can be implemented in all four collision detection models.

\begin{theorem}\label{thm-rand-ub}
There is an algorithm that, with probability $1-1/\poly(n)$, solves \ApproximateCounting{}  in $n^{o(1)}$ time with energy cost $O(\log^\ast n)$ if the model is \weak\ or \nocd, or $O(\log (\log^\ast n))$ if the model is \strong{} or \recv.
\end{theorem}


\subsection{Verifying the Correctness of an Estimate}
\label{sect:checksize}
In this section, we show how to use a dense \Census{} algorithm to verify whether a given estimate $\tilde{n}$ of network size is correct.
Suppose that there are $n$ devices agreeing on a number $\tilde{n}$.
We present an algorithm $\Verify(\tilde{n})$ that is able to check whether $\tilde{n}$ is a good estimate of $n$.
 We require that (i) a leader is elected if $n/1.5 \leq \tilde{n} \leq 1.5 n$, and (ii) no leader is elected if $\tilde{n} \geq 1.9  n$ or $\tilde{n} \leq n/1.9$.\footnote{In general, the constants $1.5, 1.9$ can both be made arbitrarily close to 1, at the cost of more time and energy.}
 The algorithm consists of two steps. The first step is to assign IDs in $[N]$ to some devices, where $N = \Theta(\log \tilde{n})$. The second step is to check whether  $\tilde{n}$ is a correct estimate via a dense \Census{} algorithm on the ID space $[N]$ with a density parameter $c$ to be determined.

\paragraph{Step~1---ID Assignment.}
We first consider the case where sender-side collision detection is available (i.e., \strong{} and \weak). We initialize $S_{\operatorname{bad}}= \emptyset$. The procedure of ID assignment consists of
$N$ time slots. For each $i \in [N]$, at the $i$th time slot each device $s \notin S_{\operatorname{bad}}$ transmits a message with probability $1/\tilde{n}$  to bid for the ID $i$.
 If a device $s$ hears back its message at the $i$th time slot, then $s$ is the only one that transmits at the $i$th time slot, and so we let $s$ assign itself the ID $i$.

 We let $\beta$ be an upper limit on the number of times a device can transmit, where $\beta$ is a sufficiently large constant.
 The purpose of setting this limit is to ensure that the energy cost is low.
 For each device $s$, if $s$ has already transmitted for $\beta$ times during the first $i$ time slots, then we add $s$ to the set $S_{\operatorname{bad}}$ at the end of the $i$th time slot, and $s$ is not allowed to transmit in future time slots during the ID assignment.

\medskip

Next, we consider the case where sender-side collision detection is not available  (i.e., \recv{} and \nocd). In this case, a transmitter does not know whether it is the only one transmitting. To resolve this issue, we increase the number of time slots from $N$ to $2 N$.

Let $i \in [N]$.
At the beginning of the $(2i-1)$th time slot, each device $s \notin S_{\operatorname{bad}}$ joins the set  $A_i$ with probability $1/\tilde{n}$,
and then each device $s \notin S_{\operatorname{bad}} \cup A_{i}$ joins the set  $B_i$ with probability $1/\tilde{n}$.

We will assign the ID $i$ to a device $s$ if $s \in A_{i}$ and $|A_i| = |B_i| = 1$.
The following procedure allows each device $s \in A_{i}$ to test if $|A_i| = |B_i| = 1$.
At the $(2i-1)$th time slot, all devices in $A_{i}$ transmit, and all devices in $B_i$ listen.
At the $(2i)$th time slot, all devices in $B_{i}$ that have successfully received a message at the $(2i-1)$th time slot transmit, and all devices in $A_i$ listen. Notice that a device $s \in A_{i}$  successfully receive a message at  the $(2i)$th time slot if and only if $|A_i| = |B_i| = 1$.

Similarly, we set $\beta$ as an upper limit on the number of times a device can join the sets $A_i$ and $B_i$, $i \in [N]$.
   Any device $s$ that has already joined these sets for  $\beta$ times is added to the set  $S_{\operatorname{bad}}$.


\medskip

Define $c = 0.325$ if the model is \strong{} or \weak; otherwise let $c = 0.325^2$.
The following lemma relates the density of the ID space to the accuracy of the estimate $\tilde{n}$.

\begin{lemma}\label{lem-network-size}
Suppose that $\tilde{n} \geq 100$.
With probability $1 - \min \{ n^{-\Omega(1)}, \tilde{n}^{-\Omega(1)}\}$, the following conditions are met:
(i)  when $\tilde{n} \geq 1.9n$ or $\tilde{n} \leq n/1.9$, either $|S_{\operatorname{bad}}| > 0$ or the number of IDs that are assigned to devices is smaller than $c N$
(ii) when $n/1.5 \leq \tilde{n} \leq 1.5 n$, we have $|S_{\operatorname{bad}}| = 0$ and the number IDs that are assigned to devices  is higher than $c N$.
\end{lemma}
\begin{proof}
We write $\mathcal{A}$ to denote the ID assignment algorithm, and write $\mathcal{A}'$ to denote a variant of the ID assignment algorithm that allows each device $s \in S_{\operatorname{bad}}$ to continue participating (i.e., there is no upper limit about the number of transmission per device). The algorithm $\mathcal{A}'$ is much easier to analyze than $\mathcal{A}$.\footnote{In the analysis of $\mathcal{A}'$, we still maintain the set $S_{\operatorname{bad}}$, but the devices in $S_{\operatorname{bad}}$ do not stop participating.} It is straightforward to see that in  $\mathcal{A}'$ the probability that an ID $i \in [N]$ is assigned is $\Prob[\mathrm{Binomial}(n, 1/ \tilde{n})=1]$ (resp., $\Prob[\mathrm{Binomial}(n, 1/ \tilde{n})=1] \cdot \Prob[\mathrm{Binomial}(n-1, 1/ \tilde{n})=1]$) when the model is  \strong{} or \weak\ (resp., \recv{} or \nocd).

We only prove the lemma for the case where the model is \strong{} or \weak; the other case is similar.
Observe that the following inequalities hold, given that  $\tilde{n} \geq 100$.
If  $\tilde{n} \geq 1.9n$ or $\tilde{n} \leq n/1.9$, then $$\Prob[\mathrm{Binomial}(n, 1/ \tilde{n})=1] < 0.32 = c - 0.005 < c.$$
If $n/1.5 \leq \tilde{n} \leq 1.5 n$, then $$\Prob[\mathrm{Binomial}(n, 1/ \tilde{n})=1] > 0.33 = c + 0.005 > c.$$

We use subscript to indicate whether a probability or an expected number refers to $\mathcal{A}$ or $\mathcal{A}'$. For instance, given an event $A$, the notation  $\Prob_{\mathcal{A}}[A]$ is the probability that $A$ occurs in an execution of $\mathcal{A}$. We write $X$ to denote the number of IDs in $[N]$ assigned to devices.
We define $\mu$ as $\Expect_{\mathcal{A}'}[X] = N \Prob[\mathrm{Binomial}(n, 1/ \tilde{n})=1]$.

\paragraph{Case 1.}
Suppose $\tilde{n} \geq 1.9 n$. We need to prove that $\Prob_{\mathcal{A}}[|S_{\operatorname{bad}}| = 0 \wedge X \geq cN] = \tilde{n}^{-\Omega(1)}$. Observe that
$$\Prob_{\mathcal{A}}[|S_{\operatorname{bad}}| = 0 \wedge X \geq cN]
= \Prob_{\mathcal{A}'}[|S_{\operatorname{bad}}| = 0 \wedge X \geq cN]
\leq  \Prob_{\mathcal{A}'}[X \geq cN].$$
Thus, it suffices to show that $\Prob_{\mathcal{A}'}[X \geq cN] = \tilde{n}^{-\Omega(1)}$.
Let $\delta = \frac{c - \Prob[\mathrm{Binomial}(n, 1/ \tilde{n})=1]}{\Prob[\mathrm{Binomial}(n, 1/ \tilde{n})=1]} > 0$.
Then $\Prob_{\mathcal{A}'}[X \geq cN] =  \Prob_{\mathcal{A}'}[X \geq (1 + \delta)\mu]$.
By a Chernoff bound, this is at most $\exp(-\delta \mu /3)$ if $\delta > 1$, and is at most $\exp(-\delta^2 \mu /3)$ if $\delta \leq 1$.
Since $c - \Prob[\mathrm{Binomial}(n, 1/ \tilde{n})=1] > 0.005$, we have: $\delta \mu = (c - \Prob[\mathrm{Binomial}(n, 1/ \tilde{n})=1]) N \geq 0.005 N$ and  $\delta^2 \mu = \frac{\left(c - \Prob[\mathrm{Binomial}(n, 1/ \tilde{n})=1]\right)^2 N}{\Prob[\mathrm{Binomial}(n, 1/ \tilde{n})=1]} \geq 0.005^2 N$. Thus, $\Prob_{\mathcal{A}'}[X \geq cN] = \exp(-\Omega(N)) = \tilde{n}^{-\Omega(1)}$.

\paragraph{Case 2.}
Suppose $n/1.5 \leq \tilde{n} \leq 1.5 n$.
We need to prove that $\Prob_{\mathcal{A}}[|S_{\operatorname{bad}}| > 0 \vee X\leq cN] = \tilde{n}^{-\Omega(1)}$.
Observe that
\begin{align*}
\Prob_{\mathcal{A}}[|S_{\operatorname{bad}}| > 0 \vee X \leq cN]
&\leq \Prob_{\mathcal{A}}[|S_{\operatorname{bad}}| > 0] + \Prob_{\mathcal{A}}[X \leq cN   \wedge |S_{\operatorname{bad}}| = 0] \\
&\leq \Prob_{\mathcal{A}'}[|S_{\operatorname{bad}}| > 0] + \Prob_{\mathcal{A}'}[X \leq cN  \wedge |S_{\operatorname{bad}}| = 0] \\
&\leq \Prob_{\mathcal{A}'}[|S_{\operatorname{bad}}| > 0] + \Prob_{\mathcal{A}'}[X  \leq cN].
\end{align*}
Thus, it suffices to show that both $\Prob_{\mathcal{A}'}[X \leq cN]$ and $\Prob_{\mathcal{A}'}[|S_{\operatorname{bad}}| > 0]$ are upper bounded by $\tilde{n}^{-\Omega(1)}$.
Let $\delta = \frac{\Prob[\mathrm{Binomial}(n, 1/ \tilde{n})=1] - c}{\Prob[\mathrm{Binomial}(n, 1/ \tilde{n})=1]} > 0$.
Then $\Prob_{\mathcal{A}'}[X \leq cN] =  \Prob_{\mathcal{A}'}[X \leq (1 - \delta)\mu]$.
By a Chernoff bound, this is at most $\exp(-\delta \mu /3)$ if $\delta > 1$, and is at most $\exp(-\delta^2 \mu /2)$ if $\delta \leq 1$.
Similarly, $\Prob[\mathrm{Binomial}(n, 1/ \tilde{n})=1] - c > 0.005$, and hence $\Prob_{\mathcal{A}'}[X \leq cN] = \exp(-\Omega(N)) = \tilde{n}^{-\Omega(1)}$.

Next, we calculate $\Prob_{\mathcal{A}'}[|S_{\operatorname{bad}}| > 0]$. The probability that a device $s$ joins $S_{\operatorname{bad}}$ in $\mathcal{A}'$ is $\Prob[\mathrm{Binomial}(N, 1/ \tilde{n}) \geq \beta] \leq N^\beta \tilde{n}^{-\beta}$. By a union bound over all $n$ devices, $\Prob_{\mathcal{A}'}[|S_{\operatorname{bad}}| > 0] \leq N^{\beta} \tilde{n}^{-\beta} n = \tilde{n}^{-\Omega(1)}$, since $N = \Theta(\log(\tilde{n}))$, $n/1.5 \leq \tilde{n} \leq 1.5 n$, and $\beta = \Omega(1)$.

\paragraph{Case 3.}
Suppose  $\tilde{n} \leq n/1.9$.  We need to prove that $\Prob_{\mathcal{A}}[|S_{\operatorname{bad}}| = 0 \wedge X \geq cN] = n^{-\Omega(1)}$. Similar to the first case, it suffices to show that $\Prob_{\mathcal{A}'}[X \geq cN] = n^{-\Omega(1)}$.
Notice that the same calculation for the first case can be applied here, and so we already have $\Prob_{\mathcal{A}'}[X\geq cN] = \tilde{n}^{-\Omega(1)}$. Thus, we only need to focus on the situation where $n$ is significantly larger than $\tilde{n}$.

Assuming  $n \geq \tilde{n}^2$,
we have
\begin{align*}
\Prob_{\mathcal{A}'}[X > 0]
&\leq N \Prob[\mathrm{Binomial}(n, 1/ \tilde{n})=1] \\
&\leq N (n / \tilde{n}) (1 - 1/\tilde{n})^{n-1}\\
&= \exp(-\Omega(n / \tilde{n}))\\
&= \exp(-\Omega(\sqrt{n})).\qedhere
\end{align*}
\end{proof}

\paragraph{Step~2---Checking the Correctness of Estimate.} We run a dense \Census{} algorithm with ID space $[N]$ and parameter $c$.
It is possible that a device is assigned to multiple IDs, and in such case the device simulates multiple devices of different IDs in the dense \Census{} algorithm. If the number of IDs that are assigned is at least $c N$, then after solving  \Census, all devices that are assigned IDs know the list of all IDs in $[N]$ that are assigned to devices.

We first consider the case where sender-side collision detection is available.
Let $s_1$ be the device that has the smallest ID.
We allocate a special time slot $t^\ast$, where $s_1$ and all devices in $S_{\operatorname{bad}}$ transmit.
The device $s_1$ elects itself as the leader if (i) $s$ has collected a list of IDs of size at least $c N$ during the \Census{} algorithm (i.e., the number of IDs that are assigned to devices is at least $c N$), and (ii) $s$ is able to hear back its message at time $t^\ast$,
i.e., $S_{\operatorname{bad}} = \emptyset$.

For the case where  sender-side collision detection is not available, $s_1$ cannot simultaneously transmit and listen.
To solve this issue, we let $s_2$ be the device that holds the smallest ID in $[N]$ excluding the ones assigned to $s_1$. Notice that a device can be assigned at most $\beta$ IDs. We let $s_2$ listen to the time slot $t^\ast$, and $s_2$ elects itself as the leader if $s_2$ hears a message from $s_1$ and the ID list resulting from the \Census{} algorithm has size at least $c N$.

\medskip

The correctness of $\Verify(\tilde{n})$ follows from Lemma~\ref{lem-network-size}. The first step costs $O(N) = O(\log \tilde{n})$ time and $O(\beta) = O(1)$ energy. The second step costs $O(N) = O(\log \tilde{n})$ time and $O(\alpha(N)) = O(\alpha(\tilde{n}))$ energy. We conclude the following lemma.

\begin{lemma}\label{lem-network-size}
With probability $1 - \min \{ n^{-\Omega(1)}, \tilde{n}^{-\Omega(1)}\}$, the algorithm $\Verify(\tilde{n})$ accomplishes the following in time $O(\log \tilde{n})$ with energy  $O(\alpha(\tilde{n}))$. A leader is elected if $n/1.5 \leq \tilde{n} \leq 1.5 n$, and no leader is elected if $\tilde{n} \geq 1.9  n$ or $\tilde{n} \leq n/1.9$.
\end{lemma}

The asymptotic time complexity of the algorithm $\Verify(\tilde{n})$ is the same as the algorithm in~\cite{BenderKPY16} which works in \strong{} and is based on circuit simulation. However, the circuit simulation takes only $O(1)$ energy while $\Verify(\tilde{n})$ needs $O(\alpha(\tilde{n}))$ energy.

\subsection{Exponential Search}
\label{sect:exp-search}
Let $D= \{d_1, d_2, \ldots\}$ be an infinite set of positive integers such that $d_{i+1} \geq \gamma \cdot {d_i}$ for each $i \geq 1$, where  $\gamma > 1$ is some large enough constant. We define $\hat{i}$ as the index such that $d_{\hat{i}-1} < \log n \leq d_{\hat{i}}$, where $n$ is the network size. For the \strong{} and the \recv{} models, we present an algorithm $\ExpSearch(D)$ that estimates $\hat{i}$ within a $\pm 1$ additive error in $O(\log \hat{i})$ time.

We first define a 1-round subroutine {\sf Test}$(i)$ as follows. Each device transmits a message with probability $2^{-d_i}$, and all other devices listen to the channel. For each listener $s$, it decides  ``$i \geq \hat{i}$'' if the channel is silent, and it decides ``$i < \hat{i}$'' otherwise.  Each transmitter decides ``$i < \hat{i}$''. It is straightforward to see that all devices make the same decision. We have the following lemma.

\begin{lemma}\label{lem-exp-aux}
Consider an execution of {\sf Test}$(i)$.
The following holds with probability $1-n^{-\Omega(1)}$.
If $i \leq \hat{i}-2$, then all devices decide ``$i < \hat{i}$''.
If $i \geq \hat{i}+1$,  then all devices decide  ``$i \geq \hat{i}$''.
\end{lemma}
\begin{proof}
Recall that $\gamma = \Omega(1)$ is chosen to be sufficiently large.
For any $i \leq \hat{i}-2$, the probability that {\sf Test}$(i)$ returns  ``$i \geq \hat{i}$'' is $\Prob[\mathrm{Binomial}(n, 2^{-d_i})=0] = n(1-2^{-d_i})^n \leq n(1-2^{\frac{\log n}{\gamma}})^n = n \cdot (1-n^{-1/\gamma})^n \leq n \cdot  \exp(-n^{1-1/\gamma}) = n^{\Omega(1)}$.  For any $i \geq \hat{i}+1$, the probability that {\sf Test}$(i)$ returns  ``$i < \hat{i}$'' is $\Prob[\mathrm{Binomial}(n, 2^{-d_i})> 0] \leq n \cdot 2^{-d_i} \leq n \cdot 2^{-\gamma\log n} = n^{-\gamma + 1} = n^{\Omega(1)}$.
\end{proof}

Based on the subroutine {\sf Test}$(i)$, the procedure  $\ExpSearch(D)$ is defined as follows.
The first step is to repeatedly run  {\sf Test}$(i)$ for $i = 1,2,4,8, \ldots$ until we reach the first index $i'$
such that all devices decide ``$i' \geq \hat{i}$'' during {\sf Test}$(i')$.
Then, we conduct a binary search using  {\sf Test}$(i)$ on the set $\{1, 2,3, \ldots, i'\}$ to find the smallest index $i$ such that {\sf Test}$(i)$ returns  ``$i \geq \hat{i}$''.
Due to Lemma~\ref{lem-exp-aux}, if all {\sf Test}$(i)$ do not fail, then it is clear that such an index $\tilde{i}$ satisfies that $\tilde{i} \in \{\hat{i}-1, \hat{i}, \hat{i}+1\}$. We conclude the following lemma.

\begin{lemma}\label{lem-exp-search}
In the \strong{} and the \recv{} models, the algorithm $\ExpSearch(D)$ finds an index $\tilde{i}$ such that $\tilde{i} \in \{\hat{i}-1, \hat{i}, \hat{i}+1\}$ in $O(\log \hat{i})$ time with probability $1-n^{-\Omega(1)}$.
\end{lemma}

\subsection{Main Algorithm}
\label{sect:main-randalg}

In this section we prove Theorem~\ref{thm-rand-ub}.
We will present an algorithm that finds an estimate $\tilde{n}$ that is within a factor of 2 of the network size $n$, i.e., $n/2 \leq \tilde{n} \leq 2n$.
Our algorithm $\EstimateSize(D)$ takes an  infinite set $D= \{d_1, d_2, \ldots\}$ of positive integers as an input parameter.
We require that $d_{i+1} \geq \gamma{d_i}$ and $d_1$ is sufficiently large such that
$\sum_{k=d_1}^\infty 1/\sqrt{2^k} \leq 1$ and $\sqrt{2^{d_1}} \geq 100$.
We will later see that different choices of $D$ lead to different time-energy tradeoffs.

With respect to the set $D$, define $\hat{i}$  as the index such that $d_{\hat{i}-1} < \log n \leq d_{\hat{i}}$. The elements in the set $D$ play the roles of ``checkpoints'' in our algorithm.  The set $D$ is independent of $n$, but $\hat{i}$ is a function of $n$.
In subsequent discussion we assume $n > 2^{d_1}$, and so the index  $\hat{i}$ is well-defined.
The reason that we are allowed to make this assumption is that for the case where $n \leq 2^{d_1} = O(1)$, we can run any \ApproximateCounting{} algorithm to fine an estimate of $n$ in $O(1)$ time.
The algorithm $\EstimateSize(D)$ is as follows.

\paragraph{Initial Setup.}
For each integer $k \geq d_1$, a device $s$ is labeled $k$ with probability $1/\sqrt{2^k}$ in such a way that $s$ is labeled by at most one number; this is the reason that we require $\sum_{k=d_1}^\infty 1/\sqrt{2^k} \leq 1$. We write $S_k$ to denote the set of all devices labeled $k$. For the case that the model is \strong\ or \recv, we do $\ExpSearch(D)$, and let $\tilde{i}$ be the result of $\ExpSearch(D)$, and set $k_0 = d_{\tilde{i}-2}$. For the case that the model is \weak\ or \nocd, set $k_0 = d_1$.

\paragraph{Finding an Estimate.}
For $k=k_0, k_0+1, k_0+2, \ldots$, do the following task.
The devices in $S_k$ collaboratively run $\Verify(\sqrt{2^k})$.
For the special case that a checkpoint is met, i.e., $k = d_i$ for some $i$, do the following additional task.
Let $L_e$ (resp., $L_o$) be the set of  leaders elected in $\Verify(\sqrt{2^{k'}})$ for all even (resp., odd) $k'$ so far (i.e., $k' \in [k_0, k]$).
We let all devices in $L_o$ simultaneously announce their labels, while all other devices listen. If exactly one message $\tilde{k}$ is sent,  the algorithm is terminated with all devices agreeing on the same estimate $\tilde{n} = 2^{\tilde{k}}$.
If the algorithm has not terminated yet, repeat the above with $L_e$.

\begin{lemma}\label{lem-label}
Define $\hat{k} = \lceil \log n \rceil$. With probability $1-\exp(\Omega(\sqrt{n}))$, the following holds.
For each $k \in [1, \hat{k} - 2]$, we have $|S_k| \geq 1.9 \sqrt{n / 2}  \geq 1.9 \sqrt{2^k}$.
For each $k \in [\hat{k}+1, \infty)$, we have $|S_k| \leq  \sqrt{2n}/1.9  \leq \sqrt{2^k} / 1.9$.
For at least one of $k \in \{\hat{k}-1, \hat{k}\}$, we have $\sqrt{2^k}/1.5 \leq |S_k| \leq 1.5\sqrt{2^k}$.
\end{lemma}
\begin{proof}
First of all, with probability $1 - n\cdot\sum_{k=n+1}^\infty 1/\sqrt{2^k} = 1 - \exp(-\Omega(n))$, no device has label greater than $n$. Therefore, in what follows we only consider the labels in the range $\{1,2, \ldots, n\}$.

 Consider the case $k \leq \hat{k} - 2$. We have $\mu = \Expect[|S_k|] = n \cdot \sqrt{2^{-k}} \geq n \cdot \sqrt{2^{-(\log n - 1)}} = \sqrt{2n}$.
Using a Chernoff bound with $\delta = 0.05$, the probability that  $|S_k| \leq 1.9 \sqrt{2^k} \leq 1.9 \sqrt{n / 2} \leq (1 - \delta)\mu$ can be upper bounded by
$\exp(-\delta^2 \mu / 2) = \exp(-\Omega(\sqrt{n}))$.

Consider the case  $n \geq k \geq \hat{k}+1$. We have $\mu = \Expect[|S_k|] = n \cdot \sqrt{2^{-k}} \leq n \cdot \sqrt{2^{-(\log n + 1)}} = \sqrt{n/2}$.
Using a Chernoff bound with $\delta = 1/1.9$, the probability that  $|S_k| \geq \sqrt{2^k}/1.9 \geq   \sqrt{2n}/1.9 \geq (1+\delta)\mu$ can be upper bounded by
$\exp(-\delta^2 \mu / 3) = \exp(-\Omega(\sqrt{n}))$.

Among the two numbers in $\{\hat{k}-1, \hat{k}\}$, we select $k \in \{\hat{k}-1, \hat{k}\}$ such that $n/\sqrt{2} \leq  2^{k} \leq \sqrt{2} n$.  Then the expected number $\mu = \Expect[|S_k|]$ satisfies  $\sqrt{2^k}/1.5 < \sqrt{2^{k}}/\sqrt{2} \leq \mu \leq \sqrt{2} \cdot \sqrt{2^{k}} < 1.5\sqrt{2^k}$. Similarly, using a Chernoff bound, we can infer that the probability that $|S_k|$ is not within  $\sqrt{2^{k}}/1.5$ and $1.5\sqrt{2^{k}}$ is at most $\exp(-\Omega(\sqrt{n}))$.
\end{proof}

\begin{lemma}\label{lem-fail-prob}
In an execution of $\EstimateSize(D)$, with probability $1-n^{-\Omega(1)}$, none of $\Verify(\sqrt{2^k})$ fails.
\end{lemma}
\begin{proof}
We assume that the statement of Lemma~\ref{lem-label} holds, since it holds with probability $1-\exp(\Omega(\sqrt{n}))$.
Similarly, with probability $1 - n\cdot\sum_{k=n+1}^\infty 1/\sqrt{2^k} = 1 - \exp(-\Omega(n))$, no device has label greater than $n$. Therefore, in what follows we only consider the labels in the range $\{1,2, \ldots, n\}$.

By Lemma~\ref{lem-network-size}, the failure  probability of $\Verify(\sqrt{2^k})$ is at most $\min\{\sqrt{2^k}^{-\Omega(1)}, |S_k|^{-\Omega(1)}\}$.
Define $\hat{k} = \lceil \log n \rceil$.
If $k \geq \hat{k}+1$, then the failure probability of $\Verify(\sqrt{2^k})$ is at most $\sqrt{2^k}^{-\Omega(1)} = n^{-\Omega(1)}$.
By Lemma~\ref{lem-label}, if $k \leq \hat{k}$, then $|S_k| = \Omega(\sqrt{n})$, and so the failure probability of $\Verify(\sqrt{2^k})$ is at most $|S_k|^{-\Omega(1)} = n^{-\Omega(1)}$.
By a union bound over all $k \in \{1, \ldots, n\}$, the probability that at least one of $\Verify(\sqrt{2^k})$  fails is bounded by $n \cdot n^{-\Omega(1)} = n^{-\Omega(1)}$.
\end{proof}

\begin{lemma}\label{thm-main-alg}
In an execution of $\EstimateSize(D)$, with probability $1-n^{-\Omega(1)}$, all devices agree on an estimate $\tilde{n}$ such that $n/2 \leq \tilde{n} \leq 2n$ in  time $T(n) = O(d_{\hat{i}}^2)$ with energy cost $E(n)$, where  $E(n) = O(\log \hat{i})$ in \strong{} and \recv, and $E(n) = O(\hat{i})$ in \weak{} and \nocd.
\end{lemma}
\begin{proof}
We assume that all of $\ExpSearch(D)$ and $\Verify(\sqrt{2^k})$, for all $k$, do not fail, since the probability that at least one of them fails is $n^{-\Omega(1)}$, in view of Lemma~\ref{lem-exp-search} and Lemma~\ref{lem-fail-prob}.
We also assume that the statement of Lemma~\ref{lem-label} holds, since it holds with probability $1-\exp(\Omega(\sqrt{n}))$.

 Define $\hat{k} = \lceil \log n \rceil$.
A consequence of Lemma~\ref{lem-label} is  that (i) there exists $k \in \{\hat{k}-1, \hat{k}\}$ such that $\Verify(\sqrt{2^k})$ elects a leader, and (ii) for each $k \notin \{\hat{k}-1, \hat{k}\}$,  $\Verify(\sqrt{2^k})$ does not elect a leader.
Recall that $\hat{i}$  is defined as the index $i$ such that $d_{{i}-1} < \log n \leq d_{{i}}$, and so  the algorithm  $\EstimateSize(D)$ must end by the iteration $k = d_{\hat{i}}$ with a correct estimate of $n$.

In what follows, we analyze the runtime and the energy cost of $\EstimateSize(D)$.
Since each $\Verify(\sqrt{2^k})$  takes  $O(k)$ time, the total time complexity is $d_{\hat{i}} \cdot O({d_{\hat{i}}}) = O(d_{\hat{i}}^2)$.

The energy cost per device in $S_k$ to make the call $\Verify(\sqrt{2^k})$ is $O(\alpha(|S_k|)) = O(\alpha(n))$, which will never be the dominant cost.
In \weak{} and \nocd, the asymptotic energy cost of $\EstimateSize(D)$ equals the number of times we encounter a checkpoint $k=d_i$ for some $d_i \in D$, which is $O(\hat{i})$.

Next, we analyze the energy cost in \strong{} and \recv. Due to $\ExpSearch(D)$ during the initial setup, the number of checkpoints encountered is reduced to $O(1)$, as we start with $k_0 = d_{\tilde{i}-2}$, where the index $\tilde{i}$ is the result of  $\ExpSearch(D)$ and satisfies $\tilde{i} \in \{\hat{i}-1, \hat{i}, \hat{i}+1\}$. Therefore, the asymptotic energy cost of $\EstimateSize(D)$ equals the energy cost of $\ExpSearch(D)$, which is $O(\log \hat{i})$.
\end{proof}

In addition to solving \ApproximateCounting, the algorithm $\EstimateSize(D)$ also solves \LeaderElection. Notice that by the end of $\EstimateSize(D)$, a unique device $s$ announces its label while all other devices listen to the channel.

\paragraph{Setting the Checkpoints.} Lemma~\ref{thm-main-alg} naturally offers a time-energy tradeoff.
We demonstrate how different choices of the checkpoints $D$ give rise to different runtime and energy cost specified in Table~\ref{table:time-energy-intro}.
For the base case, the first checkpoint $d_1$ is always chosen as a large enough constant so as to meet the three conditions: $d_{i+1} \geq \gamma{d_i}$,  $\sum_{k=d_1}^\infty 1/\sqrt{2^k} \leq 1$, and $\sqrt{2^{d_1}} \geq 100$.
In subsequent discussion we only focus on how we define $d_{i}$ inductively.

To obtain $O(\log^2 n)$ runtime, we set $d_i = \gamma d_{i-1}$ for some constant $\gamma$.
Recall that $\hat{i}$  is defined as the index $i$ such that $d_{{i}-1} < \log n \leq d_{{i}}$, and so $d_{\hat{i}} \leq \gamma \log n$. Thus, the runtime is $O(d_{\hat{i}}^2) = O(\log^2 n)$.
 With such checkpoints, the energy cost in \weak{} and \nocd{} is $O(\hat{i}) = O(\log \log n)$; the energy cost in \strong{} and \recv{} is $O(\log \hat{i}) = O(\log \log \log n)$.

For $0 < \epsilon \leq O(1)$. To obtain $O(\log^{2+\epsilon} n)$ runtime,  we set $d_i = d_{i-1}^{1+\epsilon / 2}$.
Notice that $d_{\hat{i}} \leq \log^{1+\epsilon / 2} n$. Thus, the runtime is $O(d_{\hat{i}}^2) = O(\log^{2+\epsilon} n)$.
With such checkpoints, the energy cost in \weak{} and \nocd{} is $O(\hat{i}) = O( \log_{1+\epsilon / 2} \log \log n) = O(\epsilon^{-1} \log \log \log n)$; the energy cost in \strong{} and \recv{} is $O(\log \hat{i}) = O(\log (\epsilon^{-1} \log \log \log n))$.

Theorem~\ref{thm-rand-ub} is proved as follows.
Setting $d_i = b^{d_{i-1}}$ for any constant $b > 1$ yields a polynomial time algorithm achieving the desired energy complexity, as $O(\hat{i}) = O(\log^\ast n)$ and $O(\log \hat{i}) = O(\log \log^\ast n)$. To obtain $n^{o(1)}$ runtime while maintaining the same asymptotic energy complexity, we can use  $d_i = 2^{2^{(\log d_{i-1})^\epsilon}}$, for some constant $0 < \epsilon < 1$. 
Since $\hat{i}$  is chosen such that $d_{{\hat{i}}-1} < \log n$,
we have $d_{\hat{i}} \leq 2^{2^{(\log \log n)^\epsilon}}$, and so the runtime is $O(d_{\hat{i}}^2) =  O(2^{2^{1 + (\log \log n)^{\epsilon}}}) = n^{o(1)}$.

\section{Conclusion and Open Problems}

In this paper we exposed two exponential separations in the energy complexity of \LeaderElection{}
on various wireless radio network models.  The upshot is that
randomized algorithms in $\{$\strong, \recv$\}$ are exponentially more efficient than those in
$\{$\weak, \nocd$\}$,
but deterministic algorithms in
$\{$\strong,\weak$\}$ are exponentially more efficient
than those in $\{$\recv,\nocd$\}$.
This exponential separation also occurs in the closely related problem of \ApproximateCounting.

There are a few intriguing problems that remain open in the context of \emph{single-hop} networks.
For example, is $\Theta(\alpha(N))$ the correct complexity of \LeaderElection\ and \Census\ for dense instances?
What is the true complexity of \ApproximateCounting?  In general it should exhibit a 3-way tradeoff between energy, time,
\emph{and} a given error probability.  Can $n$ anonymous devices assign themselves IDs in $\{1,\ldots,n\}$ with
$o(\log\log n)$ energy~\cite{NakanoO00} in the worst case?

Little is known about the energy-complexity of fundamental
graph problems in arbitrary (multi-hop) networks. Recently, Chang et al.~\cite{ChangDHHLP17} studied the energy complexity for \emph{broadcasting}
in multi-hop networks. It is an interesting future work direction to investigate the energy complexity for other fundamental graph problems.

\paragraph{Acknowledgement.} We would like to thank Tomasz Jurdzi\'{n}ski for discussing
the details of \cite{JurdzinskiKZ02,JurdzinskiKZ02b,JurdzinskiKZ02c,JurdziskiKZ03}
and to Calvin Newport for assisting us with the radio network literature.

\nocite{FinemanGKN16,FinemanNW16}

\bibliographystyle{plain}
\bibliography{radnet}


\end{document}